\documentclass[10pt]{article}
\usepackage[utf8x]{inputenc}
\usepackage[margin=1in]{geometry}
\usepackage{amssymb,amsmath,amsthm,amsfonts}
\usepackage{mathtools}
\usepackage{bbm}
\usepackage{natbib}
\usepackage{graphicx}
\usepackage{verbatim}
\usepackage{subcaption}
\usepackage{url}
\newtheorem{theorem}{Theorem}

\newcommand{\one}{\mathbbm{1}}
\newcommand{\R}{\mathbbm{R}}
\newcommand{\E}{\mathbbm{E}}
\newcommand{\PP}{\mathbbm{P}}

\graphicspath{{Figures/}}

\title{Assessing the calibration of multivariate probabilistic forecasts}
\author{Sam Allen \qquad Johanna Ziegel \qquad David Ginsbourger \\[1ex]{\normalsize Institute of Mathematical Statistics and Actuarial Science}\\ {\normalsize University of Bern}\\{\normalsize Bern, Switzerland}\\ {\normalsize\url{{sam.allen,johanna.ziegel,david.ginsbourger}@stat.unibe.ch}}}
\date{}

\begin{document}
	
\maketitle

\begin{abstract}     
	Rank and PIT histograms are established tools to assess the calibration of probabilistic forecasts. They not only check whether an ensemble forecast is calibrated, but they also reveal what systematic biases (if any) are present in the forecasts. Several extensions of rank histograms have been proposed to evaluate the calibration of probabilistic forecasts for multivariate outcomes. These extensions introduce a so-called pre-rank function that condenses the multivariate forecasts and observations into univariate objects, from which a standard rank histogram can be produced. Existing pre-rank functions typically aim to preserve as much information as possible when condensing the multivariate forecasts and observations into univariate objects. Although this is sensible when conducting statistical tests for multivariate calibration, it can hinder the interpretation of the resulting histograms. In this paper, we demonstrate that there are few restrictions on the choice of pre-rank function, meaning forecasters can choose a pre-rank function depending on what information they want to extract concerning forecast performance. We introduce the concept of simple pre-rank functions, and provide examples that can be used to assess the location, scale, and dependence structure of multivariate probabilistic forecasts, as well as pre-rank functions that could be useful when evaluating probabilistic spatial field forecasts. The simple pre-rank functions that we introduce are easy to interpret, easy to implement, and they deliberately provide complementary information, meaning several pre-rank functions can be employed to achieve a more complete understanding of multivariate forecast performance. We then discuss how e-values can be employed to formally test for multivariate calibration over time. This is demonstrated in an application to wind speed forecasting using the EUPPBench post-processing benchmark data set.
\end{abstract}

\section{Introduction}

It is standard practice for operational weather centres to issue forecasts that are probabilistic. Such forecasts typically take the form of an ensemble of possible weather scenarios, allowing weather centres to quantify the uncertainty inherent in their predictions. However, despite the unequivocal utility of ensemble forecasting, there is no guarantee that the issued ensemble forecasts are reliable, or calibrated, in the sense that they align statistically with the corresponding observations. To determine whether or not forecasts can be trusted, methods are required to analyse forecast calibration.

\bigskip

Although several notions of forecast calibration exist for real-valued outcomes \citep[see e.g.][]{Gneiting2007, Gneiting2021}, it is common in practice to assess whether forecasts are \textit{probabilistically calibrated}. Probabilistic calibration of ensemble predictions can be visualised using rank histograms \citep[also called Talagrand diagrams; ][]{Anderson1996, Talagrand1997, Hamill1997}. Given a set of ensemble forecasts and observations, rank histograms display the ranks of the observations when pooled among the corresponding ensemble members, thereby assessing whether the observations and the ensemble members are exchangeable. These graphical diagnostic tools are useful in practice since they not only assess calibration, but they also reveal what deficiencies (if any) are present in the forecasts. As such, rank histograms have become an integral component of probabilistic weather forecast evaluation.

\bigskip

Several extensions of rank histograms have been proposed to evaluate the calibration of probabilistic forecasts for multivariate outcomes. These multivariate rank histograms introduce a so-called \textit{pre-rank function} that condenses the multivariate forecasts and observations into univariate objects, from which a standard rank histogram can be constructed \citep{Gneiting2008}. Proposed extensions differ in the choice of pre-rank function: \cite{Smith2004} and \cite{Wilks2004} suggested using the minimum spanning tree of the multivariate ensemble members and observation as a pre-rank function; \cite{Gneiting2008} introduced a pre-rank function based on a multivariate ranking of the ensemble members and observation; \cite{Thorarinsdottir2016} proposed two alternative approaches that leverage the average rank of the observation across the individual dimensions; while \cite{Knuppel2022} recently argued that proper scoring rules are designed to summarise the information contained in the multivariate forecast and observation into a single value, and they therefore provide a canonical choice for a pre-rank function when testing for multivariate calibration. 

\bigskip

Several of these choices have been reviewed and compared using simulated forecasts and observations \citep{Thorarinsdottir2018, Wilks2017}. These comparisons illustrate how the interpretation of the multivariate rank histogram depends on the choice of pre-rank function, and that different pre-rank functions are more adept at identifying different types of mis-calibration in the forecasts. The authors of these studies therefore recommend that multiple pre-rank functions are used to construct multivariate rank histograms, to obtain a more complete understanding of how the multivariate forecasts behave.

\bigskip

However, despite the attention they have received in the literature, multivariate rank histograms are relatively rarely employed in practice. For example, weather centres regularly issue ensemble forecast fields over relevant spatial domains, and, although these forecast fields are inherently multivariate, their calibration is rarely assessed beyond evaluating univariate calibration at the individual locations or grid points. In the univariate case, the rank histogram shows the relative position of the observation among the ensemble members. In the multivariate case, existing pre-rank functions generate multivariate rank histograms with different interpretations, and practitioners may be uncertain as to which pre-rank function(s) they should apply when evaluating their forecasts. The pre-rank function is typically designed to preserve as much information as possible when condensing the multivariate forecasts and observations into univariate objects, leading to a multivariate rank histogram with a less intuitive interpretation. In particular, forecast systems with contrasting biases can often result in a similarly-shaped histogram.

\bigskip

In this paper, we argue that the main purpose of rank histograms is to identify the deficiencies in probabilistic forecasts by providing a graphical visualisation of forecast calibration. In order to achieve this effectively in a multivariate context, it is imperative that the pre-rank function is straightforward to interpret. We generalise and apply the arguments of \citet[Section 2.3]{Gneiting2008} and \citet{Ziegel2017} that any function that transforms multivariate vectors to univariate values can be used as a pre-rank function, which facilitates flexible assessments of forecast calibration. The forecaster can choose the pre-rank functions depending on what information they want to extract from their forecasts. For example, if interest is on the dependence structure of the multivariate forecast distribution, then a pre-rank function could be chosen that quantifies this dependence structure; if interest is on extreme events, then the pre-rank function could quantify how extreme the forecast or observation is, and so on. 

\bigskip

We can formally test whether or not a prediction system is calibrated by checking whether its (multivariate) rank histogram is flat. \cite{Wilks2019} recently compared approaches to achieve this using popular measures of histogram flatness. Here, we advocate an alternative approach based on e-values \citep{Arnold2021}, which provide a dynamic alternative to p-values when conducting statistical hypothesis tests. While classical tests require that the evaluation period is fixed in advance, e-values generate statistical tests that are valid sequentially. This makes them particularly relevant in sequential forecasting settings \citep[see also][]{Henzi2022}, allowing us to test forecast calibration sequentially over time without compromising type-I-error guarantees. We additionally discuss appropriate methods to address the problem of multiple testing that arises when several pre-rank functions are used to assess multivariate calibration.

\bigskip

In the following section, we introduce rank histograms, both in a univariate and multivariate setting. Section \ref{section:pre-ranks} discusses existing pre-rank functions that have been proposed in the literature, and introduces possible alternatives that assess particular features of multivariate forecasts. These pre-rank functions are employed in a simulation study in Section \ref{section:simstudy}. Section \ref{section:evalues} outlines how e-values can be employed to sequentially test for forecast calibration, and Section \ref{section:application} presents a case study in which multivariate rank histograms and e-values are used to assess the calibration of gridded wind speed ensemble forecasts over Europe. Section \ref{section:conclusion} concludes. \verb!R! code to implement the proposed multivariate histograms in practice is available at \url{https://github.com/sallen12/MultivCalibration}.

\bigskip

\section{Rank histograms}
\label{section:histograms}

In this paper, we restrict attention to ensemble forecasts, since multivariate weather forecasts are almost exclusively in this form. However, the proposed framework readily applies to continuous forecast distributions, and this is treated in detail in the Appendix. An ensemble forecast with $M$ members is a collection of possible scenarios $\mathbf{x}_{1}, \dots, \mathbf{x}_{M} \in \R^d$ for the future outcome $\mathbf{y} = \mathbf{x}_0 \in \R^d$. An ensemble forecast can be interpreted as a multivariate probabilistic forecast by considering the empirical distribution of $\mathbf{x}_{1}, \dots, \mathbf{x}_{M}$ as the predictive distribution. 

\bigskip

To assess the calibration of ensemble forecasts for real-valued outcomes ($d = 1$), we can record the rank of each observation $\mathbf{y} \in \R$ among the corresponding ensemble members $\mathbf{x}_{1}, \dots, \mathbf{x}_{M} \in \R$ for a large number of forecast cases, and check whether all ranks occur with the same frequency (up to sampling variation). This is typically achieved by displaying the distribution of the ranks in a histogram \citep{Anderson1996,Talagrand1997,Hamill1997}. Formally, the (randomised) rank of a real number $z_{0} \in \R$ amongst $z_{0}, \dots , z_{M} \in \R$ is defined as
\[
\mathrm{rank}(z_{0}; z_{1}, \dots, z_{M}) = 1 + \sum_{i=1}^{M} \one\{z_{i} < z_{0}\} + W \in \{1, \dots, M + 1\},
\]
where $\one$ denotes the indicator function, and $W$ is zero if $N = \#\{i = 1, \dots, m \mid z_{0} = z_{i}\}$ is zero, and is uniformly distributed on $\{1, \dots, N \}$ otherwise.

\bigskip

In the univariate case, an ensemble forecast is probabilistically calibrated if its rank histogram is flat. If the rank histogram is not flat, then its shape often provides additional information regarding how the forecasts are mis-calibrated: a $\cup$-shaped histogram suggests the observations are frequently either above or below all ensemble members, implying the forecasts are under-dispersed; a $\cap$-shaped histogram implies that the forecasts are over-dispersed; and a triangular histogram suggests that the forecasts tend to either over- or under-predict the outcome, indicative of a systematic forecast bias. Due to this straightforward interpretation of the rank histogram's shape, they are now well-established when evaluating operational weather forecasts. 

\bigskip

Several attempts have been made to emulate this behaviour when assessing the calibration of multivariate ensemble forecasts. Since the notion of a rank is not well-defined on $\R^{d}$, it is customary to introduce a so-called pre-rank function 
\[
\rho: \R^{d}\times \underbrace{\R^d \times \dots\times \R^d}_{\text{$M$ times}} \to \R
\]
that is invariant under permutations of the last $M$-arguments. The pre-rank function converts the multivariate observations and ensemble members to univariate objects, from which a standard rank histogram can be constructed \citep{Gneiting2008}. That is, the calibration of the multivariate ensemble forecast $\mathbf{x}_{1}, \dots, \mathbf{x}_{M}$ can be assessed by considering the univariate rank of the transformed observation $\rho(\mathbf{y},\mathbf{x}_{1}, \dots, \mathbf{x}_{M})$ within the transformed ensemble members $\rho(\mathbf{x}_{1},\mathbf{y},\mathbf{x}_{2}, \dots,\mathbf{x}_M), \dots, \rho(\mathbf{x}_{M}, \mathbf{y}, \mathbf{x}_{1}, \dots, \mathbf{x}_{M-1})$. A multivariate rank histogram is obtained by repeating this for a large number of forecast cases, and displaying the relative frequency that $\rho(\mathbf{y},\mathbf{x}_{1}, \dots, \mathbf{x}_{M})$ assumes each possible rank. If the ensemble members and the observation are exchangeable, the resulting rank histogram is uniform, and we say that the multivariate ensemble forecasts are calibrated with respect to the pre-rank function $\rho$. 

\bigskip

The simplest way to construct an admissible pre-rank function is to choose $\rho$ such that it does not depend on the last $M$ arguments. We term such pre-rank functions \textit{simple} pre-rank functions, and omit the last $M$ arguments in their definition
\[
\rho : \R^{d} \to \R.
\]
Simple pre-rank functions therefore transform the multivariate forecasts and observations to univariate summary statistics. The calibration of multivariate ensemble forecasts can analogously be assessed by ranking the transformed observation $\rho(\mathbf{y})$ within the transformed ensemble members $\rho(\mathbf{x}_{1}), \dots, \rho(\mathbf{x}_{M})$, and displaying the resulting ranks within a histogram. This approach can also readily be adapted to assess the calibration of forecasts on an arbitrary outcome space $\Omega$ by using a pre-rank function $\rho: \Omega \to \R$. While most pre-rank functions introduced in the literature are not simple, we argue that simple pre-rank functions are often more intuitive in practice, since they can easily be designed to focus evaluation on specific aspects of the multivariate forecasts. 

\bigskip

\section{Pre-rank functions}\label{section:pre-ranks}

\subsection{Existing pre-rank functions}\label{section:existing_pre-ranks}

To construct a canonical extension of the univariate rank histogram, \cite{Gneiting2008} introduced a multivariate rank as a pre-rank function:
\begin{equation*}
	\rho_{mv}(\mathbf{x}_0,\mathbf{x}_{1},\dots,\mathbf{x}_{M}) = \sum_{m=0}^{M} \one\{ \mathbf{x}_{m} \preceq \mathbf{x}_{0} \},
\end{equation*}
where $\mathbf{x}_{m} \preceq \mathbf{x}_{0}$ signifies that $x_{m,j} \leq x_{0,j}$ for all $j = 1, \dots, d$ with $\mathbf{x}_{m} = (x_{m,1}, \dots, x_{m,d}) \in \R^{d}$ for $m = 0, \dots, M$. For example, a pre-rank of $M+1$ for $\mathbf{y} = \mathbf{x}_{0}$ corresponds to a vector that exceeds the elements of the ensemble $\mathbf{x_1},\dots,\mathbf{x}_{M}$ in all dimensions, whereas a low pre-rank suggests that $\mathbf{y} = \mathbf{x}_{0}$ is not larger than the ensemble members in all dimensions. In high dimensions, it is often the case that few elements of $\mathbf{x_0},\dots,\mathbf{x}_{M}$ are comparable with respect to $\preceq$, resulting in most elements receiving the same pre-rank. Randomisation of these pre-ranks then trivially yields a flat rank histogram.

\bigskip

Recognising this, \cite{Thorarinsdottir2016} introduced two alternative pre-rank functions that are more robust to the dimensionality. The average rank pre-rank function is defined as
\begin{equation*}
	\rho_{av}(\mathbf{x}_0,\mathbf{x}_1,\dots,\mathbf{x}_M) = \frac{1}{d} \sum_{j=1}^{d} \mathrm{rank}(x_{0,j};x_{1,j},\dots,x_{M,j}).
\end{equation*}
The interpretation of the average rank is similar to the multivariate rank, but it typically leads to fewer ties between elements of $S$ and therefore requires less randomisation, making it more applicable in higher dimensions. \cite{Thorarinsdottir2016} also proposed the band-depth pre-rank function:
\begin{equation*}
	\rho_{bd}(\mathbf{x}_0,\mathbf{x}_1,\dots,\mathbf{x}_M) = \frac{1}{d} \sum_{j=1}^{d} \left[ M + 1 - \mathrm{rank}(x_{0,j};x_{1,j},\dots,x_{M,j}) \right] \left[ \mathrm{rank}(x_{0,j};x_{1,j},\dots,x_{M,j}) - 1 \right]. 
\end{equation*}
This representation assumes that there are no ties between $x_{0,j}, x_{1,j}, \dots, x_{M,j}$ for $j = 1, \dots, d$, though a more general formula exists for when this assumption does not hold \citep{Thorarinsdottir2016}. The band-depth pre-rank function measures the centrality of $\mathbf{x}_{0}$ among $\mathbf{x}_{0}, \mathbf{x}_{1}, \dots, \mathbf{x}_{M}$, with a large pre-rank suggesting that $\mathbf{x}_{0}$ is close to the centre, and a small pre-rank indicating that $\mathbf{x}_{0}$ is an outlying scenario. 

\bigskip

\cite{Smith2004} and \cite{Wilks2004} alternatively proposed using the inverse length of the minimum spanning tree of the set $\mathbf{x}_{1}, \dots, \mathbf{x}_{M}$ as a pre-rank function for $\mathbf{x}_{0}$. This pre-rank function also measures the centrality of $\mathbf{x}_0$, resulting in multivariate rank histograms with a similar interpretation to those generated using the band-depth pre-rank function. For concision, this minimum spanning tree approach is omitted from the applications in the following sections.

\bigskip

Finally, \cite{Knuppel2022} recently introduced pre-rank functions based on proper scoring rules. Proper scoring rules are functions that take a probabilistic forecast and an observation as inputs, and output a single numerical value that quantifies the forecast accuracy, thereby allowing competing forecast systems to be ranked and compared \citep[e.g.][]{GneitingRaftery2007}. \cite{Knuppel2022} argue that, by design, proper scoring rules condense the information contained in the forecasts and observations into a single value, and they therefore provide an appealing choice of pre-rank function when assessing multivariate forecast calibration. For example, a pre-rank function can be derived using the energy score, arguably the most popular scoring rule when evaluating probabilistic forecasts for multivariate outcomes \citep{GneitingRaftery2007}:
\begin{equation}\label{eq:es_pr}
	\rho_{es}(\mathbf{x}_{0}, \mathbf{x}_{1}, \dots, \mathbf{x}_{M}) = \frac{1}{M} \sum_{m=1}^{M} \| \mathbf{x}_{m} - \mathbf{x}_{0} \| - \frac{1}{2 M^{2}} \sum_{m=1}^{M} \sum_{k=1}^{M} \| \mathbf{x}_{m} - \mathbf{x}_{k} \|,
\end{equation}
where $\| \cdot \|$ denotes the Euclidean distance in $\R^{d}$. This pre-rank function measures the distance between $\mathbf{x}_{0}$ and the ensemble members $\mathbf{x}_{1}, \dots, \mathbf{x}_{M}$. A low pre-rank therefore indicates that $\mathbf{x}_{0}$ is similar to the ensemble members, whereas outlying values will receive higher pre-ranks. As \cite{Knuppel2022} remark, the latter term in this pre-rank function does not depend on $\mathbf{x}_{0}$, and could therefore be removed without changing the resulting ranks. Alternative multivariate scoring rules could also readily be used in place of the energy score.

\bigskip

\subsection{Generic pre-rank functions}\label{section:new_pre-ranks}

The pre-rank functions listed above all depend non-trivially on both $\mathbf{x}_{0}$ and $\mathbf{x}_{1}, \dots, \mathbf{x}_{M}$. That is, the function that condenses the multivariate observations and ensemble members into univariate objects depends itself on these observations and forecasts (hence the term ``pre-rank''). This is not problematic if the pre-rank functions are invariant to permutations of $\mathbf{x}_{0}, \mathbf{x}_{1}, \dots, \mathbf{x}_{M}$. In this section, we argue that it is often more intuitive to employ simple pre-rank functions that depend only on $\mathbf{x}_{0}$. We demonstrate how this approach allows us to target particular aspects of the multivariate forecasts when assessing calibration.

\bigskip

As is standard when evaluating multivariate forecasts, many of the pre-rank functions discussed herein require that the different dimensions are on the same scale. If this cannot be assumed, then the forecasts and observations should be standardised prior to evaluation. This standardisation becomes part of the pre-rank function, so it should either depend entirely on past data or be permutation invariant. For example, one can use the past climatological mean and standard deviation along each dimension, or the mean and standard deviation of the observation and ensemble members. The interpretation of the multivariate rank histograms will generally change depending on whether one has standardised or not, since dimensions on a larger scale will typically have more influence on the results; we are then checking calibration with respect to a different pre-rank function. 

\bigskip

If the goal is to visualise the (mis-)calibration of the multivariate forecasts, then the most important aspect of the pre-rank function is that it leads to multivariate rank histograms that are interpretable. Forecasters should therefore employ pre-rank functions that are capable of extracting the most relevant information from their forecasts. This does not rule out the pre-rank functions introduced above, but it highlights that practitioners need not limit themselves to these approaches when assessing multivariate calibration in practice.

\bigskip

On the other hand, if the goal is to conduct a formal statistical test for multivariate calibration, in the sense that the ensemble members and the observation are exchangeable, interpretability of pre-rank functions may not be the most important aspect. We do not consider tests for a global null hypothesis of multivariate calibration in this paper, but we comment on how such tests can be constructed in principle in Section \ref{section:evalues} based on several pre-rank functions and e-values.

\bigskip

When evaluating multivariate forecasts, it is common to focus on the location, the scale, and the dependence structure of the forecasts. Hence, rather than choosing one pre-rank function that simultaneously tries to assimilate these different elements, we propose separate pre-rank functions that assess each aspect individually; it has repeatedly been acknowledged that several pre-rank functions should be used to obtain a more complete understanding of multivariate forecast performance, so it makes sense that these pre-rank functions focus on distinct features of the multivariate forecasts. 

\bigskip

For example, a simple pre-rank function to assess the forecast's location might take the mean of the $d$ elements in the multivariate vector:
\begin{equation*}
	\rho_{loc}(\mathbf{x}_{0}) = \bar{\mathbf{x}}_{0} = \frac{1}{d} \sum_{j=1}^{d} x_{0,j}.
\end{equation*}
This pre-rank function can readily be applied to the observed outcome and each ensemble member, and the rank of the transformed observation can then be calculated. This pre-rank is easier to calculate than the existing approaches listed above. The interpretation of the resulting rank histogram is simple: if $\rho_{loc}(\mathbf{y})$ frequently attains a high (low) rank among the pre-ranks of the corresponding ensemble members, then the multivariate histogram will appear triangular, suggesting that the multivariate forecasts are negatively (positively) biased when predicting the mean, or location, of the observed vector. As in the univariate case, a $\cup$-shaped ($\cap$-shaped) histogram suggests that the ensemble forecasts are under-dispersed (over-dispersed) when predicting the observed mean.

\bigskip

Similarly, if we want to assess how well our probabilistic forecasts can predict the scale, or dispersion of the multivariate observations, we could take the variance or standard deviation of the $d$ elements as a simple pre-rank function:
\begin{equation*}
	\rho_{sc}(\mathbf{x}_{0}) = s_{\mathbf{x}_{0}}^{2} = \frac{1}{d} \sum_{j=1}^{d} \left( x_{0,j} - \bar{\mathbf{x}}_{0} \right) ^{2}.
\end{equation*}
This is valid even if the number of dimensions is small: we are not using this to estimate some unknown population variance, but rather as a simple measure of dispersion in the multivariate vector $\mathbf{x}_{0}$. The interpretation of the resulting multivariate rank histogram is analogous to above, but with the spread of the forecasts over the multivariate domain as the target variable, rather than the mean.

\bigskip

Assessing the dependence structure is less trivial, since most measures of dependence, such as correlation coefficients, are estimated from a time series of multivariate observations. The pre-rank function, on the other hand, should take only a single realisation $\mathbf{x}_{0} \in \R^{d}$ as an input. However, tools do exist to quantify the dependence in multivariate vectors. For example, the variogram at a chosen lag $h \in \{1, \dots, d - 1\}$ measures the variation between dimensions separated by lag $h$, with a small value suggesting strong dependence between these dimensions. A variogram-based pre-rank function can therefore be derived to measure the dependence between different dimensions, such as 
\begin{equation}\label{eq:variogram}
	\rho_{dep}(\mathbf{x}_{0}; h) = -\frac{\gamma_{\mathbf{x}_{0}}(h)}{s_{\mathbf{x}_{0}}^{2}},
\end{equation}
where 
\begin{equation*}
	\gamma_{\mathbf{x}_{0}}(h) = \frac{1}{2 (d - h)} \sum_{j = 1}^{d - h} |x_{0,j} - x_{0,j+h}|^{2}
\end{equation*}
is an empirical variogram at lag $h$ that only requires knowledge of the multivariate vector $\mathbf{x}_{0}$. The negative sign ensures that a larger value of $\rho_{dep}(\mathbf{x}_{0}; h)$ indicates a larger dependence, in keeping with the interpretation of $\rho_{loc}$ and $\rho_{sc}$ above. Hence, if $\rho_{dep}(\mathbf{y}; h)$ is consistently large compared to $\rho_{dep}(\mathbf{x}_{1}; h), \dots, \rho_{dep}(\mathbf{x}_{M}; h)$, then the ensemble members under-estimate the dependencies between dimensions separated by lag $h$, whereas the opposite suggests that the ensemble forecasts over-estimate the dependence at lag $h$. By measuring dependence using the empirical variogram, we implicitly assume that the lag between dimensions is meaningful, such as in a time series setting. This can be generalised further using adapted notions of spatial distances and bins \citep[e.g.][]{Gaetan2010}.

\bigskip

This pre-rank function depends on the choice of a lag $h$ at which to calculate the empirical variogram. \cite{Scheuerer2015} construct a variogram-based proper scoring rule based on the quantity
\[
\sum_{i=1}^{d} \sum_{j=1}^{d} w_{i, j} |x_{0,i} - x_{0,j}|^{2},
\]
for nonnegative weights $w_{i,j}$, arguing that it provides an effective means to evaluate a multivariate forecast's dependence structure. This quantity could also be used as a dependence pre-rank function, circumventing the choice of a specific lag $h$. However, it would require selecting a matrix of nonnegative weights. If these weights are all the same, then this weighted combination will be proportional to the variance used to define $\rho_{sc}$ above. This reflects the intrinsic difficulty in distinguishing between errors in the forecast spread and in the forecast dependence structure from only a single multivariate observation: if variables are highly dependent on each other, then the spread of the vector should be small. Dividing the empirical variogram by $s_{\mathbf{x}_{0}}^{2}$ means this dependence pre-rank function focuses more on the dependence structure and is less sensitive to errors in the scale. 

\bigskip

\subsection{Pre-rank functions for spatial fields}\label{section:spatial_pre-ranks}

The three pre-rank functions in the previous section are simple examples that condense the multivariate forecasts and observations into univariate summary statistics. Alternative pre-rank functions can also be employed in practice and we illustrate this by considering spatial field forecasts. Here, a spatial field is an element $\mathbf{x}$ in a domain $\R^{p \times q}$, where each point on the $p \times q$ grid represents a separate dimension; for ease of notation, we write $\mathbf{x}$ in place of $\mathbf{x}_{0}$ in this section. The location and scale pre-rank functions can readily be applied to spatial fields by `unravelling' $\mathbf{x}$ to obtain a vector of length $d = p \times q$. To employ the dependence pre-rank function, however, we must use an empirical variogram that is a function of spatial lags (see below). We discuss additional pre-rank functions that are tailored to the evaluation of spatial field forecasts. 

\bigskip

One such example was introduced recently by \cite{Scheuerer2018}, who used the fraction of threshold exceedances (FTE) as a pre-rank function when assessing probabilistic precipitation fields. The FTE measures the proportion of values in the spatial field that exceed some threshold of interest, and the corresponding pre-rank function is
\begin{equation}\label{eq:FTE}
	\rho_{FTE}(\mathbf{x}; t) = \frac{1}{p q} \sum_{i=1}^{p} \sum_{j=1}^{q} \one \{ x_{i,j} > t \},
\end{equation}
for some $t \in \R$. The FTE is well-known in the context of spatial forecast evaluation, since it forms the basis of the fraction skill score \citep{Roberts2008}. By employing it as a pre-rank function, we can assess to what extent the forecasts reliably predict in how many dimensions the threshold $t$ will be exceeded, giving us an idea of forecast calibration when predicting more extreme outcomes. If all components of all ensemble members and the outcome are below $t$, which will often be the case for very large thresholds, then the corresponding pre-ranks will be all be zero and the resulting rank of the observation will be determined at random. As discussed by \cite{Scheuerer2018}, these pre-ranks contain no information, meaning these instances can be omitted from the rank histogram without changing its interpretation. The FTE pre-rank function can also be applied to multivariate vectors $\mathbf{x} \in \R^{d}$, by replacing the double summation in Equation \ref{eq:FTE} with a summation over $j$ from $1$ to $d$.

\bigskip

Alternative, well-established geostatistical principles could similarly be leveraged within this framework to assess other characteristics of multivariate forecast distributions. For example, while the empirical variogram captures the variation between elements of the multivariate vector that are separated by a given lag, another aspect to consider is whether this variation changes depending on the direction between the elements. That is, whether or not the variogram is isotropic. A variogram is said to be isotropic if it depends only on the distance between elements of the multivariate vector, and not on the direction between them. By introducing a pre-rank function that measures the isotropy of an empirical variogram, we can assess to what extent the multivariate ensemble forecasts reproduce the degree of (an)isotropy present in the observed outcomes.

\bigskip

Several statistical tests have been proposed to assess whether or not the assumption of isotropy is valid when modelling spatial fields \citep[see][for a review]{Weller2016}. Here, we propose a pre-rank function inspired by the test statistic discussed in \cite{Guan2004}. Let $\mathcal{I} = \{1, \dots, p\} \times \{1, \dots, q \}$ represent a set of grid points, and define the empirical variogram of a field $\mathbf{x} \in \R^{p \times q}$ at multivariate lag $\mathbf{h} \in \{0, \dots, p - 1\} \times \{0, \dots, q - 1\}$ as
\begin{equation}
	\gamma_{\mathbf{x}}(\mathbf{h}) = \frac{1}{2|\mathcal{I}(\mathbf{h})|} \sum_{\mathbf{j} \in \mathcal{I}(\mathbf{h})} |x_{\mathbf{j}} - x_{\mathbf{j}+\mathbf{h}}|^{2},
\end{equation}
where $\mathcal{I}(\mathbf{h}) = \{\mathbf{j} \in \mathcal{I} : \mathbf{j} + \mathbf{h} \in \mathcal{I} \}$, meaning the sum is over all dimensions (i.e. grid points) that are separated by the multivariate vector $\mathbf{h}$. Equation \ref{eq:variogram} can readily be adapted to use this spatial variogram.

\bigskip

The isotropy pre-rank function considered here is then 
\begin{equation}\label{eq:isotropy}
	\rho_{iso}(\mathbf{x}; h) = - \left\{ \left[ \frac{\gamma_{\mathbf{x}}((h, 0)) - \gamma_{\mathbf{x}}((0, h))}{\gamma_{\mathbf{x}}((h, 0)) + \gamma_{\mathbf{x}}((0, h))}  \right]^{2} + \left[ \frac{\gamma_{\mathbf{x}}((h, h)) - \gamma_{\mathbf{x}}((-h, h))}{\gamma_{\mathbf{x}}((h, h)) + \gamma_{\mathbf{x}}((-h, h))}  \right]^{2} \right\}.
\end{equation}
This pre-rank function quantifies the squared distance between the variogram in the horizontal direction $\mathbf{h} = (h, 0)$ and the vertical direction $\mathbf{h} = (0, h)$, plus the squared distance between the variogram in the two diagonal directions $\mathbf{h} = (h, h)$ and $\mathbf{h} = (-h, h)$. This thus considers two separate orthogonal directions, and determines how much the variogram changes between these directions. Again, the minus sign just ensures that a higher pre-rank indicates a higher measure of isotropy. A pre-rank close to zero therefore suggests that the variogram does not change between the pairs of directions, whereas a lower value of this pre-rank function suggests that the variogram depends on the direction between elements of the field, indicative of anisotropy. 

\bigskip

The denominators in Equation \ref{eq:isotropy} help to standardise the variogram differences, recognising that the variogram will typically be larger as the lag increases. In statistical tests, this standardisation is often performed using estimates of the covariance matrix of the variogram at different lags, typically obtained using expensive resampling methods. This standardisation ensures that the squared difference between variograms at different lags asymptotically follows a multivariate normal distribution, facilitating the introduction of a chi-squared test statistic with which to test for isotropy \citep{Guan2004}. However, this is not required when defining a pre-rank function, and hence we choose an alternative standardisation that is interpretable and straightforward to implement in practice. 

\bigskip

It is most common to use unit lags (i.e. $h = 1$) within tests for isotropy \citep{Weller2016}, though alternative pairs of lags could also be employed in Equation \ref{eq:isotropy}. It is straightforward to extend this pre-rank function so that it considers multiple lags simultaneously, for example by summing $\rho_{iso}(\mathbf{x}; h)$ over different values of $h$. As with the variogram-based pre-rank function in Equation \ref{eq:variogram}, this avoids the selection of a single lag, but typically introduces additional hyperparameters. Several alternative summary measures also exist to describe the isotropy, such as the anisotropy ratio \citep[Section 2.5.2]{Chiles2012}. These measures are typically more complicated to estimate than Equation \ref{eq:isotropy}, often requiring some assumptions about the underlying data generating process, though they may also provide informative pre-rank functions. 

\bigskip

\section{Simulation study}
\label{section:simstudy}

\subsection{Multivariate Gaussian}\label{section:simstudy_mvn}

We revisit the simulation study of \cite{Thorarinsdottir2018} to illustrate how the multivariate rank histograms behave when there are errors in the location, scale, and correlation structure of the forecast distributions. Suppose the observations are drawn from a multivariate normal distribution with mean vector $\mathbf{\mu} = \mathbf{0}$, and covariance matrix $\Sigma$ for which
\begin{equation*}
	\Sigma_{i, j} = \sigma^{2} \hspace{0.05cm} \exp\left( - \frac{|i - j|}{\tau} \right), \quad i, j = 1, \dots, d.
\end{equation*}
The parameter $\sigma^{2} > 0$ controls the variance of the observations along each dimension, while $\tau > 0$ determines how quickly the correlation decays as the distance between the dimensions increases. In this sense, there is assumed to be an ordering of the variables, as is typically the case in a time series or spatial setting. We set $d = 10$, $\sigma^{2} = 1$, and $\tau = 1$. Analogous conclusions are also drawn from other configurations.

\bigskip

For each observation, $M = 20$ ensemble members are drawn at random from a mis-specified multivariate normal distribution. We consider six possible mis-specifications, corresponding to under- and over-estimation of the mean vector $\mathbf{\mu}$, scale parameter $\sigma^{2}$, and correlation parameter $\tau$. The parameter configurations corresponding to these six scenarios are listed in Table \ref{tab:mvn_ss_scens}. Corresponding multivariate rank histograms are displayed in Figure \ref{fig:mvn_ss_hists}.  

\begin{table}
	\centering
	\begin{tabular}{ l | c c c  }
		Description & $\mathbf{\mu}$ & $\sigma^{2}$ & $\tau$ \\
		\hline
		Under-estimation of mean & (-0.5, -0.5, -0.5) & 1 & 1 \\
		Over-estimation of mean & (0.5, 0.5, 0.5) & 1 & 1 \\
		Under-estimation of variance & $\mathbf{0}$ & 0.85 & 1 \\
		Over-estimation of variance & $\mathbf{0}$ & 1.25 & 1 \\
		Under-estimation of correlation & $\mathbf{0}$ & 1 & 0.5 \\
		Over-estimation of correlation & $\mathbf{0}$ & 1 & 2
	\end{tabular}
	\caption{Mis-specifications considered in the simulated ensemble forecasts.} 
	\label{tab:mvn_ss_scens}
\end{table}

\begin{figure}[t!]
	\centering
	\begin{subfigure}{0.16\linewidth}
		\includegraphics[width=\linewidth]{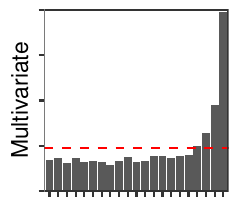}
		\includegraphics[width=\linewidth]{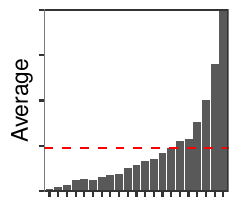}
		\includegraphics[width=\linewidth]{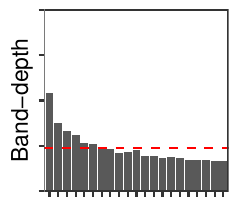}
		\includegraphics[width=\linewidth]{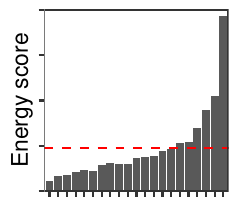}
		\includegraphics[width=\linewidth]{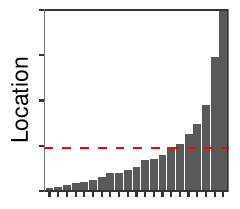}
		\includegraphics[width=\linewidth]{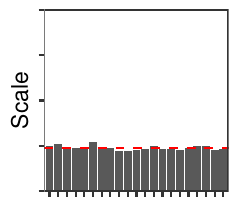}
		\includegraphics[width=\linewidth]{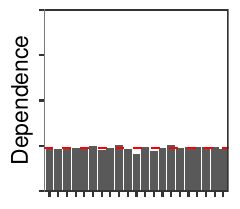}
		\caption{}
	\end{subfigure}
	\begin{subfigure}{0.16\linewidth}
		\includegraphics[width=\linewidth]{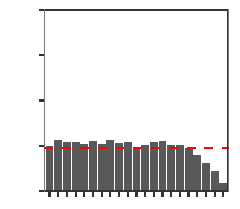}
		\includegraphics[width=\linewidth]{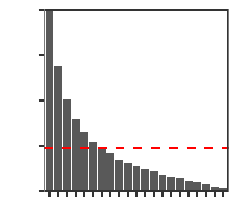}
		\includegraphics[width=\linewidth]{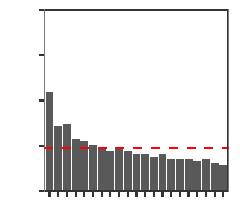}
		\includegraphics[width=\linewidth]{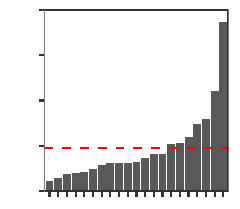}
		\includegraphics[width=\linewidth]{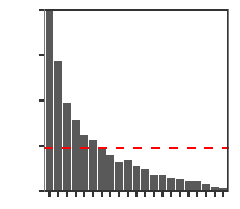}
		\includegraphics[width=\linewidth]{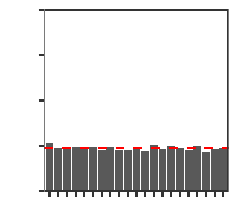}
		\includegraphics[width=\linewidth]{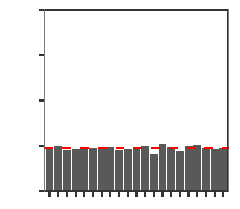}   
		\caption{}
	\end{subfigure}
	\begin{subfigure}{0.16\linewidth}
		\includegraphics[width=\linewidth]{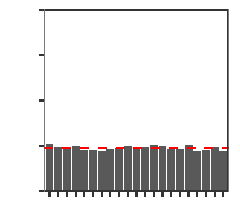}
		\includegraphics[width=\linewidth]{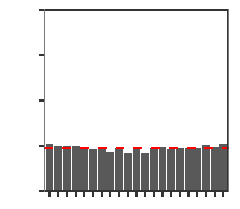}
		\includegraphics[width=\linewidth]{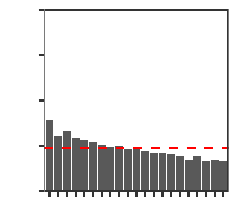}
		\includegraphics[width=\linewidth]{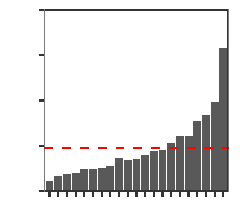}
		\includegraphics[width=\linewidth]{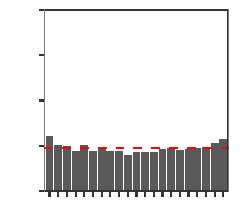}
		\includegraphics[width=\linewidth]{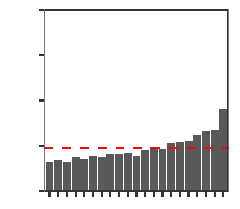}
		\includegraphics[width=\linewidth]{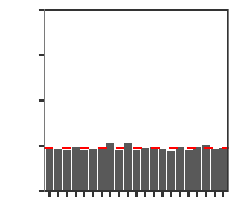}
		\caption{}
	\end{subfigure}
	\begin{subfigure}{0.16\linewidth}
		\includegraphics[width=\linewidth]{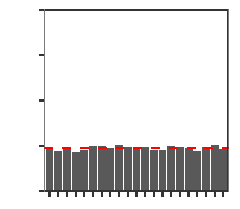}
		\includegraphics[width=\linewidth]{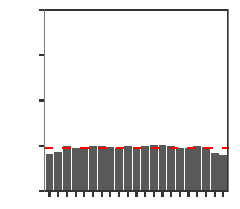}
		\includegraphics[width=\linewidth]{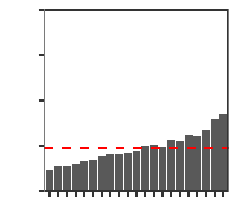}
		\includegraphics[width=\linewidth]{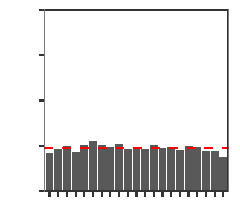}
		\includegraphics[width=\linewidth]{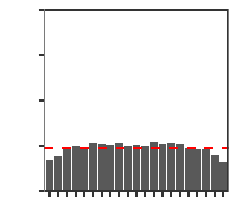}
		\includegraphics[width=\linewidth]{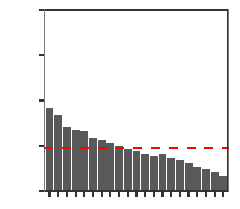}
		\includegraphics[width=\linewidth]{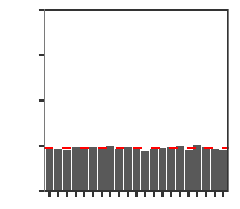}
		\caption{}
	\end{subfigure}
	\begin{subfigure}{0.16\linewidth}
		\includegraphics[width=\linewidth]{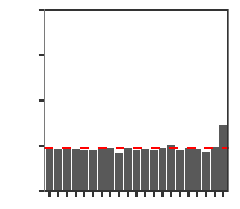}
		\includegraphics[width=\linewidth]{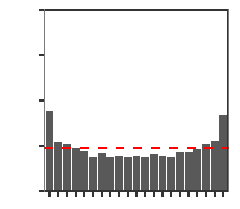}
		\includegraphics[width=\linewidth]{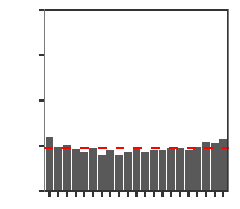}
		\includegraphics[width=\linewidth]{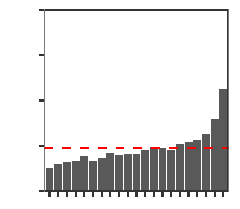}
		\includegraphics[width=\linewidth]{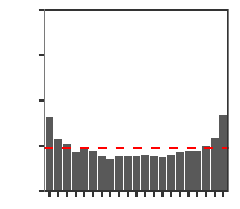}
		\includegraphics[width=\linewidth]{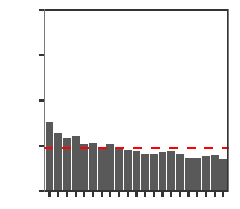}
		\includegraphics[width=\linewidth]{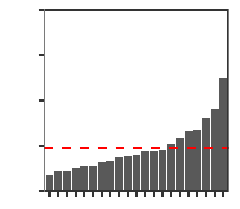}
		\caption{}
	\end{subfigure}
	\begin{subfigure}{0.16\linewidth}
		\includegraphics[width=\linewidth]{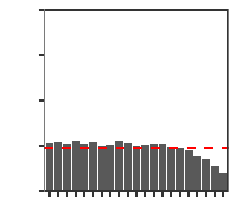}
		\includegraphics[width=\linewidth]{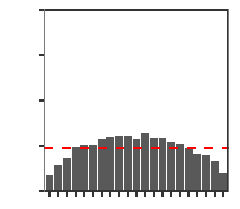}
		\includegraphics[width=\linewidth]{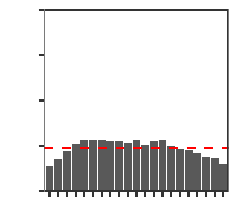}
		\includegraphics[width=\linewidth]{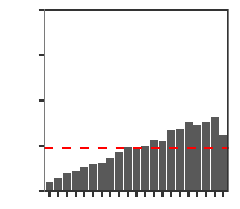}
		\includegraphics[width=\linewidth]{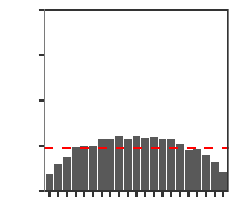}
		\includegraphics[width=\linewidth]{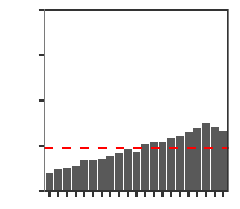}
		\includegraphics[width=\linewidth]{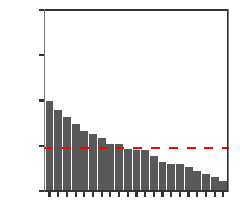}
		\caption{}
	\end{subfigure}
	\caption{Results from the Multivariate Gaussian simulation study. Multivariate rank histograms constructed using the seven pre-rank functions discussed in the text when the mean is (a) under-estimated: $\mu = (-0.5, -0.5, -0.5)$, (b) over-estimated: $\mu = (0.5, 0.5, 0.5)$; when the variance is (c) under-estimated: $\sigma^{2} = 0.85$, (d) over-estimated: $\sigma^{2} = 1.25$; when the correlations are (e) under-estimated: $\tau = 0.5$, (f) over-estimated: $\tau = 2$. The histograms were constructed from 10,000 observations and ensemble forecasts.}
	\label{fig:mvn_ss_hists}
\end{figure}

\bigskip

The multivariate and average pre-rank functions are insensitive to changes in the scale, whereas the band-depth pre-rank function, in measuring the centrality of the observation among the ensemble members, distinguishes well between under- and over-dispersed forecast distributions. On the other hand, when errors are present in the mean of the forecast distributions, the band-depth pre-rank function is unable to differentiate between under-prediction and over-prediction, in contrast to the multivariate and average ranks.

\bigskip

As \cite{Knuppel2022} remark, the energy score pre-rank function also quantifies the centrality of the observation among the ensemble members. However, in almost all cases, the energy score pre-rank function results in a multivariate rank histogram of the same shape, with the observation receiving a higher pre-rank than the ensemble members. In this case, we are comparing the energy score obtained by the ensemble forecast when the observation $\mathbf{y}$ occurs, to the score obtained when each of the ensemble members $\mathbf{x}_{m}$ ($m = 1, \dots, M$) occurs as the observation; it is not too surprising that the score is generally largest when the observation is not a member of the ensemble, resulting in these negatively skewed histograms. Hence, although this pre-rank function leads to powerful tests for multivariate calibration, the corresponding rank histograms are unable to distinguish between different types of forecast errors. This could perhaps be circumvented by treating the observation as an additional ensemble member when calculating Equation \ref{eq:es_pr}, in which case this pre-rank function should behave similarly to the band-depth pre-rank function.

\bigskip

Since none of the existing pre-rank functions can distinguish between all three types of forecast error, \cite{Wilks2017} and \cite{Thorarinsdottir2018} recommend that several pre-rank functions are implemented. This motivates the use of specific pre-rank functions that are designed to focus on each aspect of forecast performance individually. As desired, the location pre-rank function clearly discriminates between errors in the mean of the forecast distribution, the scale pre-rank function between errors in the scale, and the dependence pre-rank function between errors in the correlation of the multivariate forecasts. The dependence pre-rank function is implemented here with lag $h=1$. The scale and dependence pre-rank functions are insensitive to biases in the mean, though the location pre-rank function also slightly detects errors in the scale and dependence structure. As discussed, while the dependence pre-rank function is insensitive to forecast dispersion errors, it is difficult to construct a pre-rank function that can identify errors in the scale of the forecast distribution, without also being insensitive to errors in the correlation structure.

\bigskip

\subsection{Gaussian random fields}\label{section:simstudy_grf}

Consider now the case where the observations are spatial fields. In particular, suppose they are realisations of a Gaussian random field on a regular $30 \times 30$ grid. We assume the grid is standardised such that the distance between adjacent grid points is one unit. This extends the previous example to a higher dimensional setting in which there is additionally spatial structure present in the data. As before, the observations are drawn from a zero-mean random field with an exponential covariance function such that the covariance between two locations $\mathbf{i}$ and $\mathbf{j}$ on the grid is
\begin{equation}\label{eq:cov_func}
	\sigma^{2} \hspace{0.05cm} \exp\left( - \frac{||\mathbf{i} - \mathbf{j}||}{\tau} \right), \quad \mathbf{i}, \mathbf{j} \in \{1, \dots, 30\} \times \{1, \dots, 30\},
\end{equation}
which is governed by a scale parameter $\sigma^{2} = 1$ and correlation parameter $\tau = 1$.

\bigskip

Ensembles of size $M=20$ are similarly drawn from a mis-specified Gaussian random field, and the calibration of these multivariate ensemble forecasts is again assessed under several possible mis-specifications. In this case, results are presented for errors in the scale, correlation, and isotropy of the forecast distributions. The multivariate rank histograms corresponding to biases in the forecast mean are similar to those in Figure \ref{fig:mvn_ss_hists} (not shown). 

\bigskip

The errors in the scale and correlation structure are analogous to those listed in Table \ref{tab:mvn_ss_scens}. The exponential covariance function used to generate the observations depends only on the distance between the different dimensions, and not the direction between them. This observation-generating process is therefore isotropic. To generate forecasts that misrepresent the isotropy in this process, geometric anisotropy is introduced by rescaling the grids in the vertical direction by a factor of 1.25. A second scenario is also considered whereby the observation fields are generated using this approach, with the forecasts obtained using the (isotropic) covariance function in Equation \ref{eq:cov_func}. 

\bigskip

For concision, the multivariate rank and energy score pre-rank functions have been omitted from this higher-dimensional simulation study: as discussed, in high dimensions, the multivariate rank trivially leads to a flat rank histogram, since most pre-ranks are determined at random; the energy score pre-rank function behaves as in the previous simulation study, and fails to distinguish between the different types of forecast errors. Results are also presented for the FTE and isotropy pre-rank functions introduced in Section \ref{section:spatial_pre-ranks}. The FTE pre-rank function is employed with a threshold of $t=1$, roughly equal to the 85th percentile of a standard normal distribution.

\bigskip

These multivariate rank histograms are displayed in Figure \ref{fig:grf_ss_hists}. The average pre-rank function fails to identify errors in the scale and isotropy, whereas the band-depth rank differentiates well between under- and over-estimation of the scale and dependence in the forecast distributions. Unsurprisingly, the location pre-rank function is largely insensitive to any of the forecast errors, since it is designed to focus on the mean of the predictive distributions. The scale pre-rank function is highly sensitive to errors in the variance and dependence structure, while the dependence pre-rank function (at unit multivariate lag) is able to detect errors in the correlation structure. The FTE pre-rank function is insensitive to the scale of the forecast distributions, though it does distinguish between errors in the correlation structure.

\bigskip

Finally, none of these pre-rank functions can identify errors in the isotropy of the forecast distributions. The measure of isotropy given in Equation \ref{eq:isotropy} appears to yield a pre-rank function that is very sensitive to these errors, and is therefore capable of informing the forecaster when the forecast fields misrepresent the isotropy of the true variogram. This is true both when the observed fields are generated from an isotropic process but the forecast fields are not, and also vice versa. This pre-rank function is insensitive to all other types of forecast errors, meaning it focuses uniquely on the isotropy.

\begin{figure}[t!]
	\centering
	\begin{subfigure}{0.16\linewidth}
		\includegraphics[width=\linewidth]{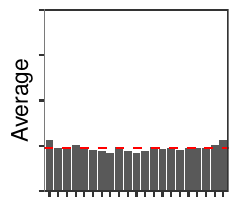}
		\includegraphics[width=\linewidth]{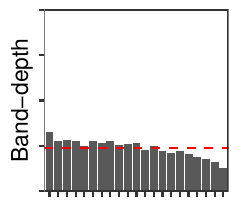}
		\includegraphics[width=\linewidth]{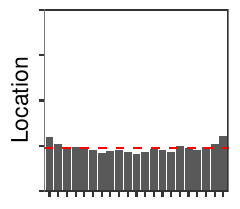}
		\includegraphics[width=\linewidth]{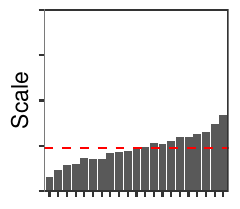}
		\includegraphics[width=\linewidth]{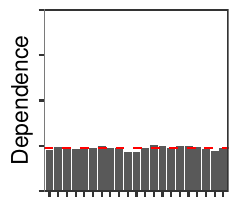}
		\includegraphics[width=\linewidth]{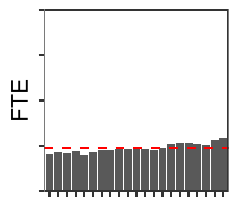}
		\includegraphics[width=\linewidth]{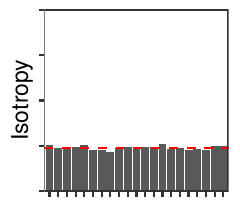}
		\caption{}
	\end{subfigure}
	\begin{subfigure}{0.16\linewidth}
		\includegraphics[width=\linewidth]{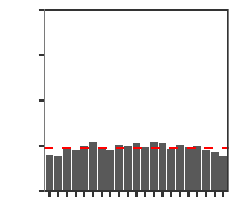}
		\includegraphics[width=\linewidth]{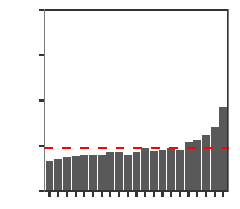}
		\includegraphics[width=\linewidth]{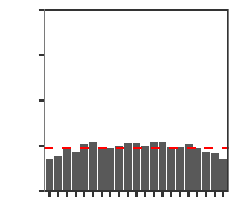}
		\includegraphics[width=\linewidth]{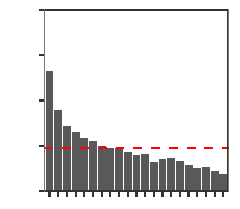}
		\includegraphics[width=\linewidth]{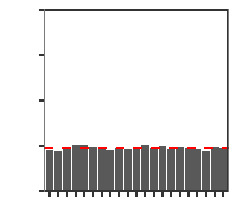}
		\includegraphics[width=\linewidth]{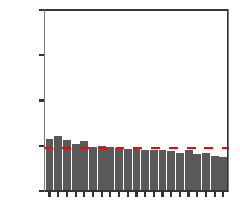}
		\includegraphics[width=\linewidth]{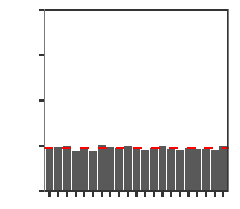}   
		\caption{}
	\end{subfigure}
	\begin{subfigure}{0.16\linewidth}
		\includegraphics[width=\linewidth]{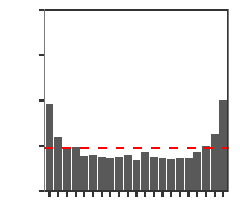}
		\includegraphics[width=\linewidth]{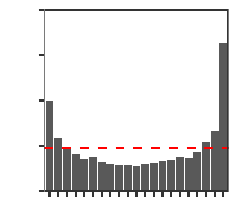}
		\includegraphics[width=\linewidth]{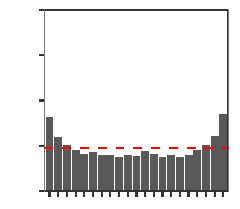}
		\includegraphics[width=\linewidth]{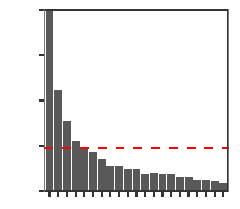}
		\includegraphics[width=\linewidth]{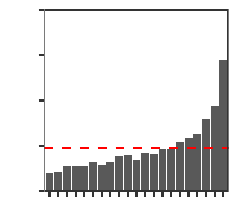}
		\includegraphics[width=\linewidth]{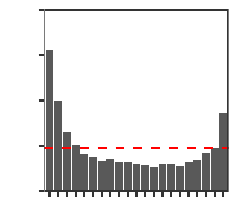}
		\includegraphics[width=\linewidth]{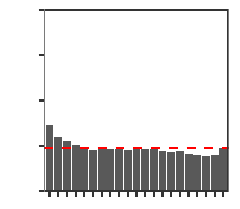}
		\caption{}
	\end{subfigure}
	\begin{subfigure}{0.16\linewidth}
		\includegraphics[width=\linewidth]{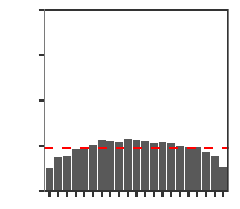}
		\includegraphics[width=\linewidth]{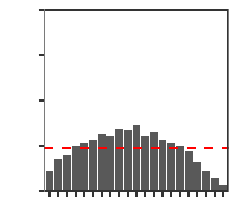}
		\includegraphics[width=\linewidth]{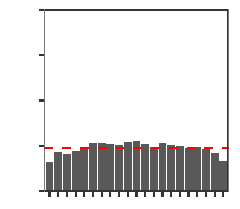}
		\includegraphics[width=\linewidth]{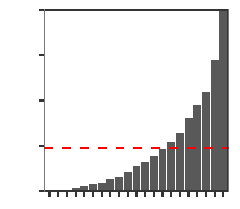}
		\includegraphics[width=\linewidth]{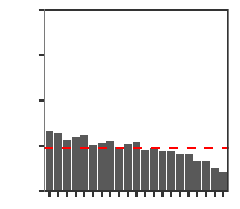}
		\includegraphics[width=\linewidth]{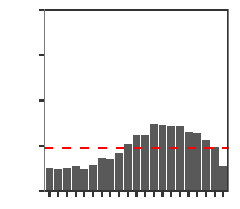}
		\includegraphics[width=\linewidth]{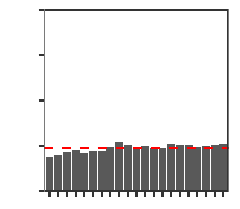}
		\caption{}
	\end{subfigure}
	\begin{subfigure}{0.16\linewidth}
		\includegraphics[width=\linewidth]{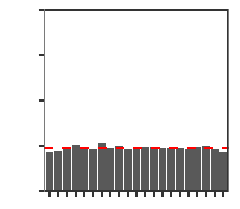}
		\includegraphics[width=\linewidth]{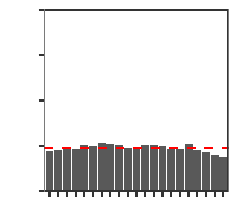}
		\includegraphics[width=\linewidth]{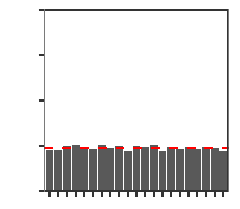}
		\includegraphics[width=\linewidth]{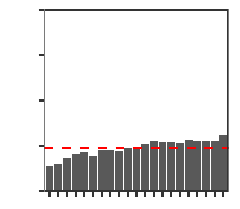}
		\includegraphics[width=\linewidth]{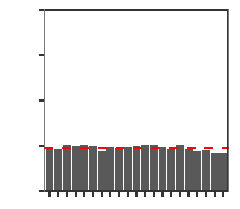}
		\includegraphics[width=\linewidth]{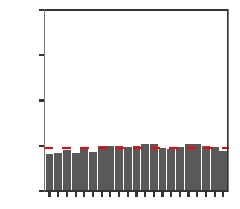}
		\includegraphics[width=\linewidth]{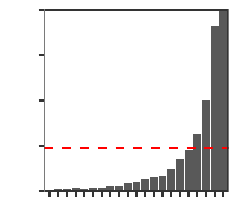}
		\caption{}
	\end{subfigure}
	\begin{subfigure}{0.16\linewidth}
		\includegraphics[width=\linewidth]{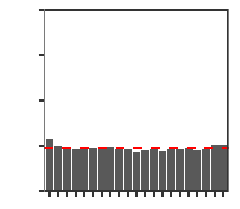}
		\includegraphics[width=\linewidth]{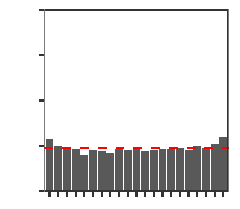}
		\includegraphics[width=\linewidth]{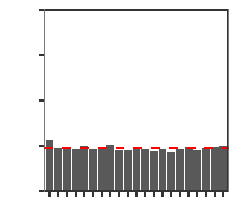}
		\includegraphics[width=\linewidth]{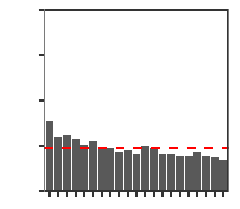}
		\includegraphics[width=\linewidth]{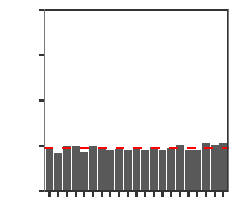}
		\includegraphics[width=\linewidth]{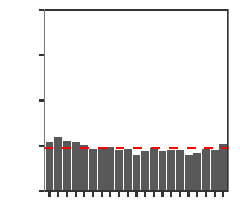}
		\includegraphics[width=\linewidth]{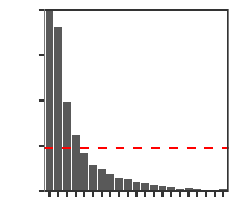}
		\caption{}
	\end{subfigure}
	\caption{Results from the Gaussian random field simulation study. Multivariate rank histograms constructed using the seven pre-rank functions discussed in the text when the variance is (a) under-estimated: $\sigma^{2} = 0.85$, (b) over-estimated: $\sigma^{2} = 1.25$; when the correlations are (c) under-estimated: $\tau = 0.5$, (d) over-estimated: $\tau = 2$; when the degree of isotropy is (e) under-estimated, (f) over-estimated. The histograms were constructed from 10,000 observations and ensemble forecasts.}
	\label{fig:grf_ss_hists}
\end{figure}

\bigskip

\section{Monitoring calibration with e-values}
\label{section:evalues}

To demonstrate the behaviour of the different pre-rank functions, the simulation studies in the previous section consist of forecasts that are independent and identically distributed. The performance of weather forecasts, however, generally varies over time: for example, in different seasons, weather regimes, or due to updates in numerical weather prediction models. Since these conditional biases may cancel each other out such that they are not visible in the rank histograms, it is additionally useful to monitor how forecast calibration evolves over time. 

\bigskip

While several metrics have been proposed to quantify the flatness of a rank histogram \citep[see][]{Wilks2019}, we advocate an approach based on e-values, which simultaneously provide a sequentially valid test for forecast calibration \citep{Arnold2021}. E-values provide a dynamic alternative to p-values when conducting statistical hypothesis tests, and while classical tests require that the evaluation period is fixed in advance, e-values generate statistical tests that are valid sequentially, i.e. under optimal stopping. An e-value for a given null hypothesis is a nonnegative random variable $E$ such that $\E[E] \leq 1$ when the null hypothesis is true. Realisations of $E$ that are less than one therefore support the null hypothesis, whereas larger e-values provide evidence against the null hypothesis. 

\bigskip

When testing for the flatness of a rank histogram, a suitable null hypothesis is that the rank of the observation is uniformly distributed on the set of possible ranks $\{1, \dots, M + 1\}$. \cite{Arnold2021} propose monitoring calibration using the e-values
\begin{equation*}
	E_{t} = (M + 1) p_{A}(R_{t}),
\end{equation*}
where $R_{t}$ is a random variable denoting the rank of the observation among the ensemble members at time $t \ge 1$, and $p_{A}(r)$ is the probability of rank $r \in \{1, \dots, M + 1\}$ occurring under the alternative hypothesis. The alternative distribution $p_{A}$ can be estimated from the ranks observed prior to the current time. This can be done empirically, but if the number of ensemble members is even moderately large, \cite{Arnold2021} suggest modelling $p_{A}$ using a beta-binomial distribution, whose parameters are estimated sequentially using maximum likelihood estimation based on the previously observed ranks.

\bigskip

Essentially, if a rank $r$ has already been observed more often than the expected frequency under the null hypothesis, then $p_{A}(r)$ should be larger than $1/(M + 1)$, resulting in a value $E_{t} > 1$. This recognises that the rank histogram is becoming less similar to a flat histogram, giving evidence against the null hypothesis. Conversely, if $r$ has previously been observed less often than expected, then $E_{t}$ will be smaller than one, recognising that by observing $r$ the rank histogram becomes closer to a flat histogram.

\bigskip

In a sequential setting, if $E_t$  is an e-value conditional on the information available at time $t-1$, for each $t \ge 1$, then the product $e_{t} = \prod_{i=1}^{t} E_{i}$ represents the cumulative evidence for or against the null hypothesis up until time $t$, and is itself an e-value. Theoretical arguments then justify rejecting the null hypothesis at time $t$ and significance level $\alpha$ if $e_{t} \geq 1/\alpha$, and this can be monitored over time without having to fix the sample size in advance. This is equivalent to rejecting the null hypothesis if $1/e_{t} \leq \alpha$; that is, the reciprocal of an e-value constitutes a conservative p-value. 

\bigskip

The calibration of forecasts with a lead time $k$ of one time unit can readily be assessed using this approach. For lead times $k > 1$, the situation is more complicated and we refer to \citet{Arnold2021} for theoretical details. In summary, we can aggregate $k$ sequences of e-values that are separated by lag $k$ using the average product
\begin{equation*}
	e_{t} = \frac{1}{k} \sum_{j=1}^{k} \prod_{i \in \mathcal{K}_{j}(t)} E_{i},
\end{equation*}
where $\mathcal{K}_{j}(t) = \{j + sk : s = 0, 1, \dots, t-1, \hspace{0.2cm} j + sk \leq t \}$ is the set of indices between $j$ and $t$ that are separated by lag $k$. \cite{Arnold2021} demonstrate that if this aggregated sequence is scaled by the constant $(e \log(k))^{-1}$ (with $e = \mathrm{exp(1)}$), then $e_{t}$ will exceed $1/\alpha$ with probability at most $\alpha$. Hence, when testing whether or not a forecast system is calibrated at lead time $k > 1$, the null hypothesis of calibration can be rejected at time $t$ and significance level $\alpha$ if $e_{t} \geq e \log(k)/\alpha$.

\bigskip

By visualising $e_{t}$ as a function of time, it becomes easy to identify periods where evidence against calibration is accumulating. Figure \ref{fig:evalues} presents examples for the forecasts considered in the following section. An increasing e-value indicates a period of mis-calibration, whereas a declining e-value suggests that either the forecasts are calibrated, or the type of mis-calibration has changed. While the e-value does not specify what type of mis-calibration is present in the forecasts, a time series of e-values can be used to identify periods worth analysing in further detail \citep[see][for further details]{Arnold2021}. 

\bigskip

E-values allow us to test whether multivariate forecasts are probabilistically calibrated with respect to a chosen pre-rank function. When several pre-rank functions are used to assess different facets of calibration, one can compute e-values for each pre-rank function, but care must be taken when hypotheses of calibration are rejected since this introduces a multiple testing problem. The simplest approach to account for multiple testing is to apply a Bonferroni correction; that is, instead of rejecting the null hypothesis when the e-value surpasses the threshold $1/\alpha$ or $e \log(k)/\alpha$, for forecasts of lag $1$ and lag $k > 1$ respectively, a hypothesis is only rejected if the e-value surpasses the threshold $\ell/\alpha$ or $\ell e \log(k)/\alpha$, respectively, where $\ell$ is the number of pre-rank functions considered. As with p-values, there are more powerful methods to correct for multiple testing than the Bonferroni correction, and an alternative procedure for e-values is proposed in \cite{Vovk2021}. 

\bigskip

If one aims to construct powerful tests for the global null hypothesis of exchangeability of the multivariate ensemble members and the observation, it appears sensible to use a limited number of pre-rank functions that provide complementary information, and to combine the resulting e-values into one e-value by predictable mixing; we refer readers to \cite{Waudby-SmithRamdas2023} and \cite{CasgrainLarssonETAL2023}, where predictable mixing of test martingales is discussed and is referred to as so-called betting strategies. The same reasoning applies in the case of general multivariate predictive distributions; here, the null hypothesis would be that the forecasts are auto-calibrated, see Appendix. We do not explore this proposal in this paper since we want to emphasise understanding deficiencies in the calibration of multivariate ensemble forecasts, and this cannot be achieved using one global test for calibration. 

\bigskip

\section{Case study}
\label{section:application}

\subsection{Data}
\label{section:data}

Section \ref{section:simstudy} demonstrates how the different pre-rank functions behave when evaluating forecasts in an idealised framework. In this section, we apply the same pre-rank functions to evaluate the calibration of gridded ensemble forecasts for European wind speeds. We consider 10m wind speed forecasts and observations taken from the European Meteorological Network's (EUMETNET) post-processing benchmark dataset \citep[EUPPBench;][]{Demaeyer2023}, which has recently been introduced to provide a ``common platform based on which different post-processing techniques of weather forecasts can be compared.'' The dataset is also canonical when illustrating forecast evaluation methods.

\bigskip

The EUPPBench dataset contains daily forecasts issued by the European Center for Medium-range Weather Forecasts' (ECMWF) Integrated Forecasting System (IFS) during 2017 and 2018. Twenty years of reforecasts are also available for every Monday and Thursday in this two year period, which are generated by the same forecast model but with a smaller number of ensemble members (11 instead of 51). We restrict attention here to the 20 years of reforecasts at a lead time of five days, meaning we work with ensembles of size $M=11$. 

\bigskip

These forecasts are compared to ERA5 reanalyses \citep{Hersbach2020}, which provide a best guess for the observed wind speed fields. The forecasts and observations are on a regular longitude-latitude grid that covers a small domain in central Europe (2.5-10.5E, 45.75-53.5N). The grid has a horizontal resolution of 0.25$^{\circ}$, roughly corresponding to 25 kilometers, and is therefore comprised of 33 distinct longitudes and 32 latitudes. This domain is displayed in Figure \ref{fig:example_fields}. 

\bigskip

When forecasting surface weather variables such as wind speed, the ensemble forecasts issued by numerical weather models are typically subject to systematic biases. To remove these biases, it is common to statistically post-process the numerical model output. Statistical post-processing methods use a historical archive of forecasts and observations to learn and then remove systematic errors in the raw ensemble forecasts. Since the goal of this paper is to illustrate the utility afforded by multivariate rank histograms, we employ a simple, well-known univariate post-processing scheme to re-calibrate the wind speed ensemble forecasts at each grid point separately, and then compare competing approaches to convert these univariately post-processed forecasts into spatially-coherent probabilistic forecast fields. Readers are referred to \cite{Vannitsem2018} for a thorough overview of statistical post-processing methods.

\bigskip

To post-process the IFS ensemble forecasts at each grid point, we employ a standard ensemble model output statistics (EMOS) approach \citep{Gneiting2005}, in which the future wind speed is assumed to follow a truncated logistic distribution \citep{Scheuerer2015b}. The predictive distributions are truncated below at zero, meaning positive probability is assigned only to positive wind speeds. The location parameter of the truncated logistic distribution is a linear function of the ensemble mean forecast at the same time and location, and the scale parameter is similarly a linear function of the ensemble standard deviation. The parameters of this post-processing model are estimated using the first 15 years of reforecasts, and the resulting predictions are then assessed using the remaining 5 years. This results in 3135 forecast and observation fields for model training, and 1045 pairs for verification.

\bigskip

This univariate post-processing method re-calibrates the wind speed forecasts at each individual grid point. However, in doing so, the multivariate dependencies present in the raw ensemble forecast are lost. To obtain forecast distributions that have a realistic multivariate dependence structure, it is common to extract evenly-spaced quantiles from the univariate post-processed distributions, and then reorder these quantiles according to a relevant dependence template. The most popular approach to achieve this is ensemble copula coupling \citep[ECC;][]{Schefzik2013}. ECC is an empirical copula approach that uses the raw ensemble forecasts issued by the numerical weather models as a template to reorder the post-processed forecast distributions. ECC is known to be a straightforward and effective approach to reinstall the dependencies present in the numerical ensemble forecasts, and is frequently implemented in operational post-processing suites.

\bigskip

For comparison, we additionally compare the resulting forecast fields to those obtained using an alternative copula-based reordering scheme, namely the Schaake Shuffle \citep{Clark2004}. Like ECC, the Schaake Shuffle reorders evenly-spaced quantiles from the post-processed forecast distributions according to some dependence template. In contrast to ECC, the Schaake Shuffle uses a random selection of past multivariate observations to construct the dependence template, rather than the raw ensemble forecast. Although the Schaake Shuffle can leverage an arbitrary number of previous observations in this dependence template, only eleven observations are used here, so that the resulting ensemble forecasts have the same number of members as those generated using ECC. Further details regarding these two approaches can be found in \cite{Lakatos2022} and references therein. Several variants of both ECC and the Schaake Shuffle have also recently been proposed, but we restrict attention here to the two most widely-used implementations; \cite{Lakatos2022} find that these extensions do not provide significant benefits. 

\bigskip

The calibration of the forecast fields obtained from these two post-processing methods is compared to the calibration of the raw numerical model output. An example of the forecast fields generated using the three methods is presented in Figure \ref{fig:example_fields}. 

\begin{figure}
	\centering
	\begin{subfigure}{0.24\linewidth}
		\includegraphics[width = \linewidth]{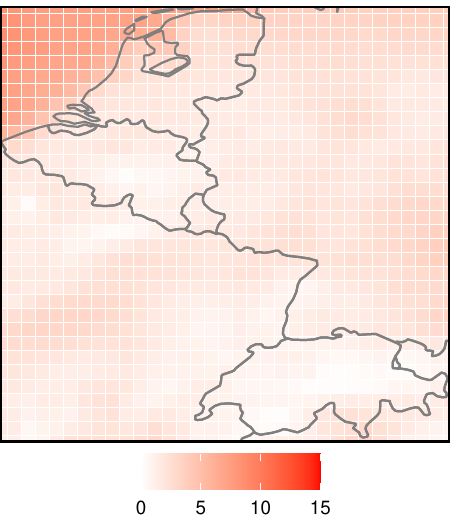}
		\caption{Observation}
	\end{subfigure}
	\begin{subfigure}{0.24\linewidth}
		\includegraphics[width = \linewidth]{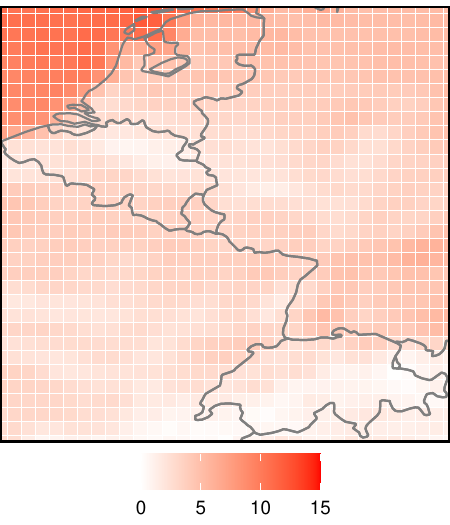}
		\caption{Raw ensemble}
	\end{subfigure}
	\begin{subfigure}{0.24\linewidth}
		\includegraphics[width = \linewidth]{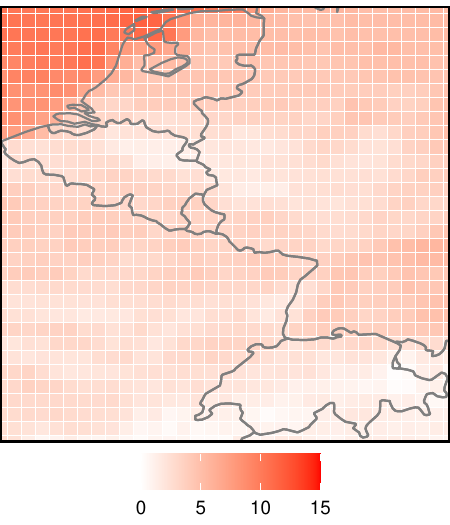}
		\caption{ECC}
	\end{subfigure}
	\begin{subfigure}{0.24\linewidth}
		\includegraphics[width = \linewidth]{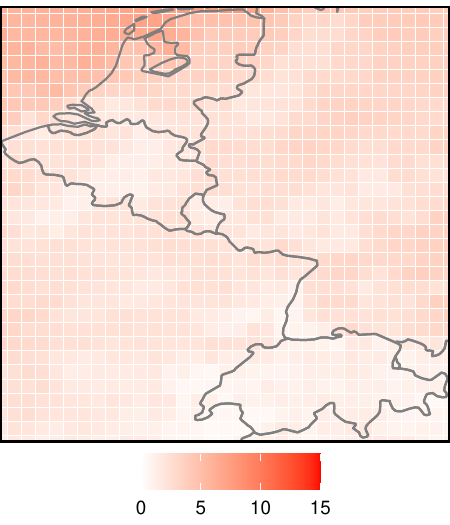}
		\caption{Schaake Shuffle}
	\end{subfigure}
	\caption{Example observation and forecast fields on a randomly selected day.}
	\label{fig:example_fields}
\end{figure}

\bigskip

\subsection{Results}

Firstly, consider the univariate calibration of the competing forecasting methods. Figure \ref{fig:univ_rank_hists} displays the univariate rank histograms corresponding to the raw ensemble forecasts (i.e. the IFS reforecasts) before and after undergoing post-processing. The ranks are aggregated across all 1056 grid points. Although we describe two different multivariate post-processing methods, these differ only in how they reorder the univariate post-processed forecasts, and thus result in the same univariate rank histograms. Figure \ref{fig:univ_rank_hists} illustrates that the raw ensemble forecasts are slightly under-dispersed, with the observed wind speed falling outside the range of the ensemble members more often than would be expected of a calibrated forecast. The simple post-processing approach corrects these dispersion errors.

\begin{figure}[b!]
	\centering
	\includegraphics[width = 0.3\linewidth]{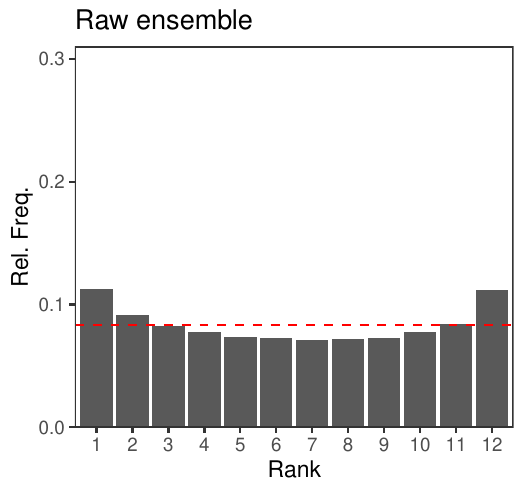}
	\includegraphics[width = 0.3\linewidth]{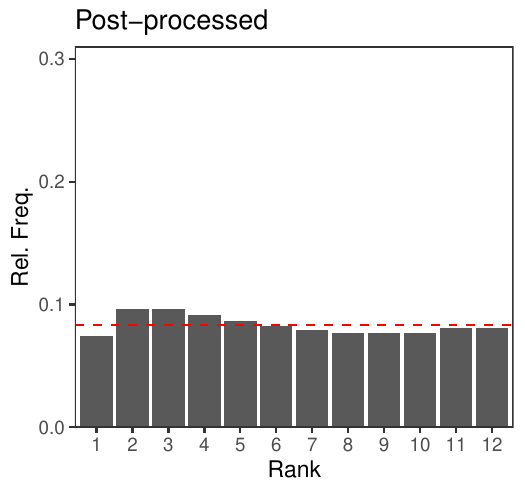}
	\caption{Univariate rank histograms corresponding to the raw ensemble forecasts before and after post-processing. The dashed red lines indicate a uniform histogram.}
	\label{fig:univ_rank_hists}
\end{figure}

\bigskip

Figure \ref{fig:mv_rank_hists} displays the corresponding multivariate rank histograms. The same pre-rank functions are employed as in Section \ref{section:simstudy_grf}. The dependence and isotropy pre-rank functions are implemented with unit lag, and the FTE pre-rank function uses a threshold $t=6$ms$^{-1}$, roughly corresponding to the 90$^{th}$ percentile of the wind speeds in the training data across all grid points. The average and location pre-rank functions suggest that the raw ensemble forecasts are relatively well-calibrated when predicting the mean wind speed over the domain. However, the corresponding band-depth histogram indicates that these forecasts do not reliably capture the centrality of the observation among the ensemble members. The IFS forecast fields also appear to reliably predict the variance of the observed wind speed fields, as well as the number of threshold exceedances, but they severely under-estimate the dependence between adjacent grid points, and over-estimate the measure of isotropy in the ERA5 reanalyses.

\begin{figure}
	\centering
	\begin{subfigure}{0.19\linewidth}
		\includegraphics[width=\linewidth]{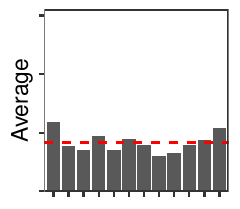}
		\includegraphics[width=\linewidth]{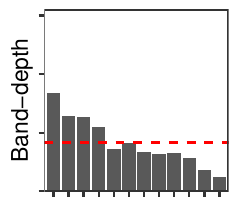}
		\includegraphics[width=\linewidth]{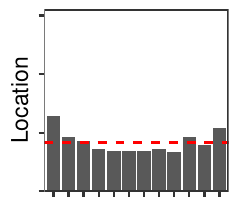}
		\includegraphics[width=\linewidth]{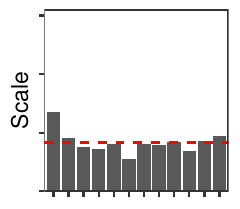}
		\includegraphics[width=\linewidth]{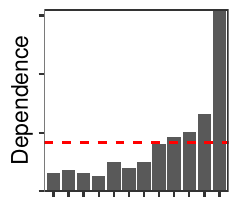}
		\includegraphics[width=\linewidth]{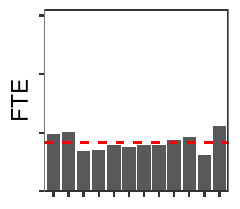}
		\includegraphics[width=\linewidth]{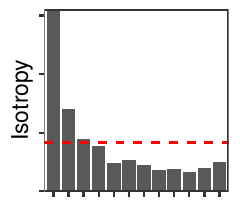}
		\caption{Raw ensemble}
	\end{subfigure}
	\begin{subfigure}{0.19\linewidth}
		\includegraphics[width=\linewidth]{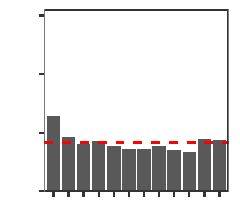}
		\includegraphics[width=\linewidth]{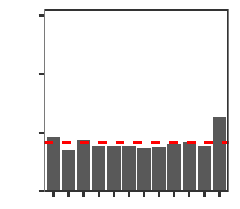}
		\includegraphics[width=\linewidth]{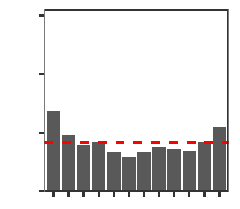}
		\includegraphics[width=\linewidth]{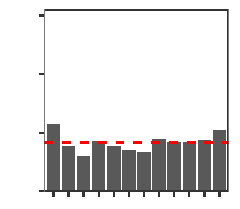}
		\includegraphics[width=\linewidth]{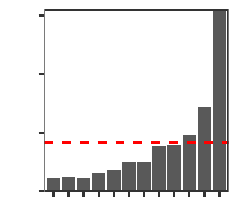}
		\includegraphics[width=\linewidth]{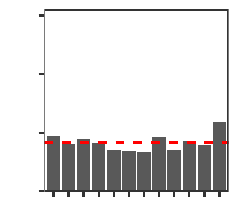}
		\includegraphics[width=\linewidth]{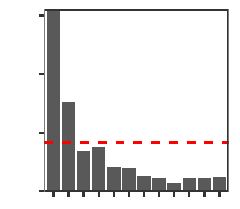}
		\caption{ECC}
	\end{subfigure}
	\begin{subfigure}{0.19\linewidth}
		\includegraphics[width=\linewidth]{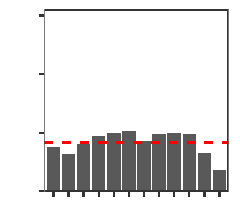}
		\includegraphics[width=\linewidth]{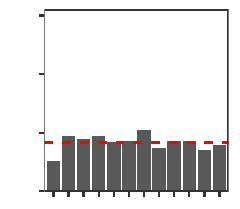}
		\includegraphics[width=\linewidth]{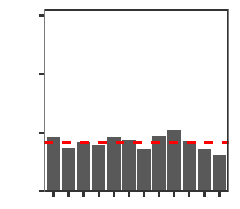}
		\includegraphics[width=\linewidth]{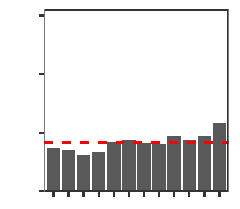}
		\includegraphics[width=\linewidth]{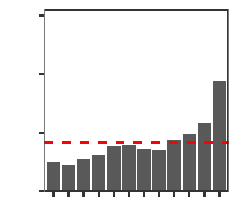}
		\includegraphics[width=\linewidth]{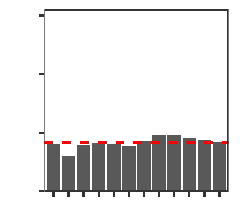}
		\includegraphics[width=\linewidth]{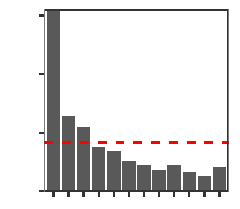}
		\caption{Schaake Shuffle}
	\end{subfigure}
	\caption{Multivariate rank histograms constructed using seven pre-rank functions for (a) the raw ensemble forecasts; the post-processed forecasts reordered using (b) ECC, (c) the Schaake Shuffle. The dashed red lines indicate a uniform histogram.}
	\label{fig:mv_rank_hists}
\end{figure}

\bigskip

Consider now the post-processed forecasts that are reordered according to ECC. The multivariate rank histogram corresponding to the band-depth pre-rank function is much closer to uniform. While this may suggest that these post-processed forecasts improve upon the multivariate calibration of the IFS output, we can use targeted pre-rank functions to identify remaining sources of mis-calibration in these forecasts. In particular, post-processing using EMOS and ECC does not correct the errors in the dependence structure that manifest in the IFS forecast fields, and the resulting forecasts also over-estimate the measure of isotropy at unit lag. Aside from the band-depth pre-rank function, the multivariate rank histograms for the ECC forecasts exhibit the same patterns as those for the raw IFS forecasts. The Schaake Shuffle performs similarly, though while these forecasts also under-estimate the dependence between adjacent grid points, they do so to a lesser degree than IFS and ECC.

\bigskip

These patterns are reinforced by e-values. Figure \ref{fig:evalues} displays the cumulative e-values when assessing the calibration of the forecast systems with respect to three of the chosen pre-rank functions. E-values are displayed for the band-depth, scale, and dependence pre-rank functions. In all cases, a ``burn-in'' period of 100 forecast-observation pairs is required from which to estimate $p_{A}$ when calculating the e-values. The three forecasting methods issue calibrated predictions of the scale of the ERA5 reanalyses fields. For the band-depth pre-rank function, there is quickly sufficient evidence to conclude that the multivariate rank histogram corresponding to the raw ensemble forecasts is not uniform. The post-processing methods, on the other hand, appear probabilistically calibrated with respect to the band-depth pre-rank function. For the dependence pre-rank function, as suggested by the multivariate rank histograms, there is sufficient evidence at the 5\% level to suggest that none of the multivariate forecasts are calibrated. Hence, although post-processing offers improvements upon the raw IFS forecasts, there is still vast potential to improve upon these baseline post-processing methods. Moreover, we expect that the mis-calibration of all forecasting methods would be more pronounced if the forecasts were compared to station observations rather than reanalysis fields.

\bigskip

Since multivariate forecast calibration is assessed using multiple pre-rank functions, one might ask what constitutes a good set of pre-rank functions? A collection of pre-rank functions will be most useful when the individual pre-rank functions provide complementary information. To assess this, it is useful to have a measure of dependence between the rank of $\rho(\mathbf{y})$ for different choices of the pre-rank function $\rho$. If the dependence is strong, then it suggests that one of the pre-rank functions provides redundant information. 

\bigskip

Table \ref{tab:correlation} contains the correlations between the ranks of $\rho(\mathbf{y})$ among $\rho(\mathbf{x}_{1}), \dots, \rho(\mathbf{x}_{M})$ for the various choices of $\rho$. Results are shown for the raw ensemble members, though very similar correlations are obtained for the two post-processing methods. There is a strong positive correlation between the average rank and location pre-rank functions, which both assess the mean behaviour of the spatial fields, and also the FTE pre-rank function. There is also strong positive correlation between the scale, dependence, and FTE pre-rank functions. The location pre-rank function is strongly correlated with the scale pre-rank function in this example, suggesting that errors when predicting the average wind speed over the spatial domain are linked to errors when predicting the variation in wind speeds over the domain. The band-depth and isotropy pre-rank functions, on the other hand, exhibit relatively low correlations with the other pre-rank functions, making them particularly useful in this application. 

\begin{figure}[t]
	\centering
	\includegraphics[width=0.32\linewidth]{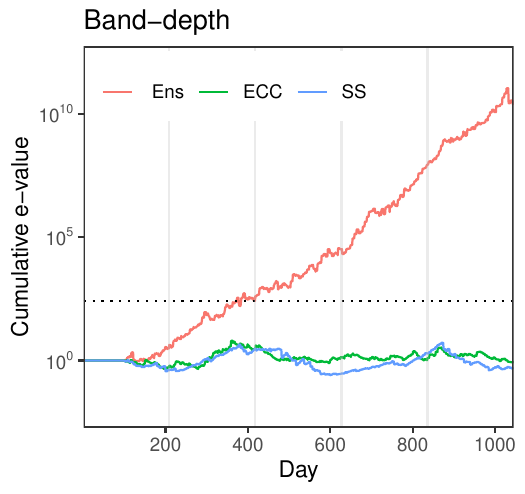}
	\includegraphics[width=0.32\linewidth]{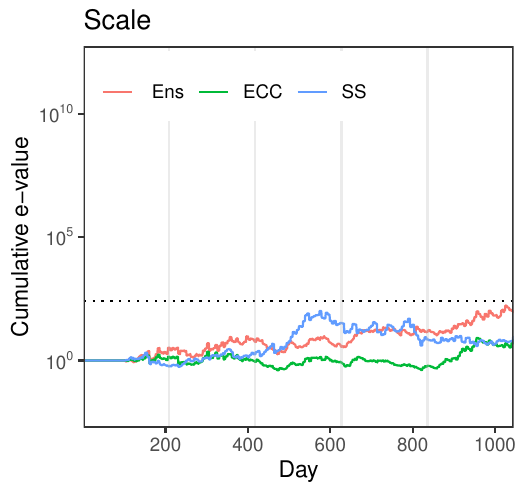}
	\includegraphics[width=0.32\linewidth]{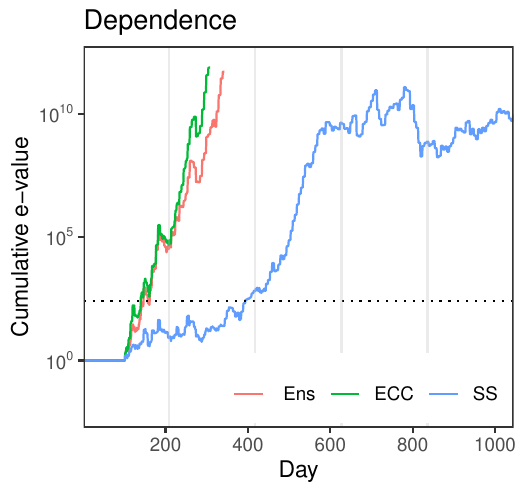}
	\caption{Cumulative e-values for the raw ensemble forecast and the post-processed forecasts when sequentially testing for calibration. Results are shown for the multivariate rank histograms constructed using the band-depth, scale, and dependence pre-rank functions. The null hypothesis of calibration with respect to a pre-rank function is rejected at the 5\% level if the cumulative e-value exceeds the dotted horizontal line at $3 e\log(5)/0.05$. The vertical grey lines divide the test data into five year periods.}
	\label{fig:evalues}
\end{figure}

\begin{table}[b]
	\centering
	\begin{tabular}{ c | c c c c c c c}
		& Average & Band-depth & Location & Scale & Dependence & FTE & Isotropy \\
		\hline
		Average & 1.00 & -0.02 & 0.95 & 0.58 & 0.35 & 0.67 & -0.22 \\
		Band-depth & -0.02 & 1.00 & -0.03 & -0.09 & 0.08 & -0.08 & -0.05 \\
		Location & 0.95 & -0.03 & 1.00 & 0.70 & 0.47 & 0.76 & -0.27 \\
		Scale & 0.58 & -0.09 & 0.70 & 1.00 & 0.73 & 0.83 & -0.35 \\
		Dependence & 0.35 & 0.08 & 0.47 & 0.73 & 1.00 & 0.63 & -0.28 \\
		FTE & 0.67 & -0.08 & 0.76 & 0.83 & 0.63 & 1.00 & -0.33 \\
		Isotropy & -0.22 & -0.05 & -0.27 & -0.35 & -0.28 & -0.33 & 1.00 \\
	\end{tabular}
	\caption{Correlation between the ranks of $\rho(\mathbf{y})$ among the transformed raw ensemble members for the various pre-rank functions.}
	\label{tab:correlation}
\end{table}

\section{Conclusion}
\label{section:conclusion}

For probabilistic forecasting systems to be as useful as possible, the forecasts they issue must be calibrated, in the sense that they are statistically consistent with what actually materialises. In practice, the calibration of univariate ensemble forecasts is typically assessed using rank histograms. Ensemble forecasts for multivariate outcomes can be evaluated using multivariate rank histograms, which use a so-called pre-rank function to transform the observations and ensemble members into univariate objects, prior to constructing a standard rank histogram. In this paper, we highlight that there is considerable flexibility in the choice of pre-rank function, meaning practitioners can choose pre-rank functions on a case-by-case basis, depending on what information they wish to extract from their forecasts. In particular, while previously proposed pre-rank functions depend on the forecast and observed outcome, we argue that this need not be the case. Instead, simple pre-rank functions can be employed that more directly target specific aspects of the multivariate forecasts. 

\bigskip

We introduce generic pre-rank functions that can focus attention to the location, scale, and dependence structure of the multivariate forecast distributions, allowing each of these components to be assessed individually. The resulting histograms are straighforward to interpret, making it easier to assess what systematic deficiencies exist in the forecasts. We also propose suitable pre-rank functions to assess the calibration of probabilistic spatial field forecasts, which are regularly issued by operational weather forecasting centres. Although we focus here on spatial forecast fields as an example, the arguments presented herein apply to other multivariate forecasts, and future work could consider relevant pre-rank functions in other multivariate forecasting settings, such as time series forecasting.

\bigskip

A long-standing challenge when evaluating spatial field forecasts is to design verification tools that value realistic-looking forecast fields, i.e. fields that do not violate any physical laws. If we could quantify how realistic a weather field is, then these measures of realism could be used as pre-rank functions in the framework discussed herein. Doing so would not only reveal whether our forecast fields are realistic, but, if not, the corresponding multivariate rank histograms should additionally allow us to identify how our forecasts are unrealistic.

\bigskip

The pre-rank functions we introduced are used to compare ensemble fields obtained from a near-operational ensemble prediction system before and after having undergone statistical post-processing. These results help to understand not only how the post-processed forecasts improve upon the raw model output, but also what deficiencies these forecasts still exhibit. In recognising the limitations of the predictions, it becomes easier to remove them, resulting in more accurate and reliable forecasts in the future. While the multivariate post-processing frameworks employed in this study are commonly applied in operational post-processing suites, the multivariate rank histograms in Figure \ref{fig:mv_rank_hists} suggest that these forecast still exhibit significant biases, particularly related to the dependence between the wind speed at nearby grid points. An interesting avenue for future work would therefore be to compare these results to those obtained using state-of-the-art machine learning models that have recently been introduced for multivariate post-processing \citep[e.g.][]{Dai2021,Chen2022}. 

\bigskip

The general framework outlined herein involves identifying univariate characteristics of multivariate objects, and evaluating the multivariate forecasts via their ability to predict these characteristics. This approach has also been proposed when constructing multivariate scoring rules \citep[see e.g.][]{Scheuerer2015, Weniger2017, Heinrich2021}. \cite{Allen2022} illustrate that this framework can essentially be interpreted as a weighting of the scoring rule, where the transformation determines which outcomes are to be emphasised when calculating forecast accuracy. Weighted versions of these multivariate rank histograms could also be employed to target particular multivariate outcomes when assessing forecast calibration, as outlined by \cite{Allen2023}. Note, however, that the transformations used within scoring rules need not be interpretable, unlike the pre-rank function used to construct multivariate rank histograms, and hence some transformations that are suitable when calculating forecast accuracy may not be when interest is on forecast calibration. 



\bigskip

\section*{Acknowledgements}

The authors are grateful to the Swiss Federal Office of Meteorology and Climatology (MeteoSwiss) and the Oeschger Centre for Climate Change Research for financially supporting this work. Sebastian Arnold, Jonas Bhend, Alexander Henzi, and Marc-Oliver Pohle are also thanked for fruitful discussions during the preparation of this manuscript.

\bibliographystyle{apalike}
\bibliography{references}

\begin{thebibliography}{}

\bibitem[Allen et~al., 2023]{Allen2023}
Allen, S., Bhend, J., Martius, O., and Ziegel, J. (2023).
\newblock Weighted verification tools to evaluate univariate and multivariate
  probabilistic forecasts for high-impact weather events.
\newblock {\em Weather and Forecasting}, 38:499–516.

\bibitem[Allen et~al., 2022]{Allen2022}
Allen, S., Ginsbourger, D., and Ziegel, J. (2022).
\newblock Evaluating forecasts for high-impact events using transformed kernel
  scores.
\newblock {\em arXiv preprint arXiv:2202.12732}.

\bibitem[Anderson, 1996]{Anderson1996}
Anderson, J.~L. (1996).
\newblock A method for producing and evaluating probabilistic forecasts from
  ensemble model integrations.
\newblock {\em Journal of Climate}, 9:1518--1530.

\bibitem[Arnold et~al., 2021]{Arnold2021}
Arnold, S., Henzi, A., and Ziegel, J.~F. (2021).
\newblock Sequentially valid tests for forecast calibration.
\newblock {\em arXiv preprint arXiv:2109.11761}.

\bibitem[Casgrain et~al., 2023]{CasgrainLarssonETAL2023}
Casgrain, P., Larsson, M., and Ziegel, J. (2023).
\newblock Anytime-valid sequential testing for elicitable functionals via
  supermartingales.
\newblock {\em Bernoulli}.
\newblock To appear. Preprint available at \url{arXiv:2204.05680}.

\bibitem[Chen et~al., 2022]{Chen2022}
Chen, J., Janke, T., Steinke, F., and Lerch, S. (2022).
\newblock Generative machine learning methods for multivariate ensemble
  post-processing.
\newblock {\em arXiv preprint arXiv:2211.01345}.

\bibitem[Chiles and Delfiner, 2012]{Chiles2012}
Chiles, J.-P. and Delfiner, P. (2012).
\newblock {\em Geostatistics: modeling spatial uncertainty}, volume 713.
\newblock John Wiley \& Sons.

\bibitem[Clark et~al., 2004]{Clark2004}
Clark, M., Gangopadhyay, S., Hay, L., Rajagopalan, B., and Wilby, R. (2004).
\newblock The {Schaake} shuffle: A method for reconstructing space--time
  variability in forecasted precipitation and temperature fields.
\newblock {\em Journal of Hydrometeorology}, 5:243--262.

\bibitem[Dai and Hemri, 2021]{Dai2021}
Dai, Y. and Hemri, S. (2021).
\newblock Spatially coherent postprocessing of cloud cover ensemble forecasts.
\newblock {\em Monthly Weather Review}, 149:3923--3937.

\bibitem[Demaeyer et~al., 2023]{Demaeyer2023}
Demaeyer, J., Lerch, S., Primo, C., Van~Schaeybroeck, B., Atencia, A.,
  Ben~Bouall{\`e}gue, Z., Chen, J., Dabernig, M., Evans, G., Faganeli~Pucer,
  J., et~al. (2023).
\newblock The {EUPPBench} postprocessing benchmark dataset v1.0.
\newblock {\em Earth System Science Data Discussions}, pages 1--25.

\bibitem[Diebold et~al., 1998]{Diebold1998}
Diebold, F.~X., Gunther, T.~A., and Tay, A.~S. (1998).
\newblock Evaluating density forecasts with applications to financial risk
  management.
\newblock {\em International Economic Review}, 39:863--883.

\bibitem[Gaetan and Guyon, 2010]{Gaetan2010}
Gaetan, C. and Guyon, X. (2010).
\newblock {\em Spatial statistics and modeling}, volume~90.
\newblock New York: Springer.

\bibitem[Gneiting et~al., 2007]{Gneiting2007}
Gneiting, T., Balabdaoui, F., and Raftery, A.~E. (2007).
\newblock Probabilistic forecasts, calibration and sharpness.
\newblock {\em Journal of the Royal Statistical Society: Series B},
  69:243--268.

\bibitem[Gneiting and Raftery, 2007]{GneitingRaftery2007}
Gneiting, T. and Raftery, A.~E. (2007).
\newblock Strictly proper scoring rules, prediction, and estimation.
\newblock {\em Journal of the American Statistical Association}, 102:359--378.

\bibitem[Gneiting et~al., 2005]{Gneiting2005}
Gneiting, T., Raftery, A.~E., Westveld~III, A.~H., and Goldman, T. (2005).
\newblock Calibrated probabilistic forecasting using ensemble model output
  statistics and minimum {CRPS} estimation.
\newblock {\em Monthly Weather Review}, 133:1098--1118.

\bibitem[Gneiting and Ranjan, 2013]{Gneiting2013}
Gneiting, T. and Ranjan, R. (2013).
\newblock Combining predictive distributions.
\newblock {\em Electronic Journal of Statistics}, 7:1747--1782.

\bibitem[Gneiting and Resin, 2021]{Gneiting2021}
Gneiting, T. and Resin, J. (2021).
\newblock Regression diagnostics meets forecast evaluation: {Conditional
  calibration}, reliability diagrams, and coefficient of determination.
\newblock {\em arXiv preprint arXiv:2108.03210}.

\bibitem[Gneiting et~al., 2008]{Gneiting2008}
Gneiting, T., Stanberry, L.~I., Grimit, E.~P., Held, L., and Johnson, N.~A.
  (2008).
\newblock Assessing probabilistic forecasts of multivariate quantities, with an
  application to ensemble predictions of surface winds.
\newblock {\em Test}, 17:211--235.

\bibitem[Guan et~al., 2004]{Guan2004}
Guan, Y., Sherman, M., and Calvin, J.~A. (2004).
\newblock A nonparametric test for spatial isotropy using subsampling.
\newblock {\em Journal of the American Statistical Association}, 99:810--821.

\bibitem[Hamill and Colucci, 1997]{Hamill1997}
Hamill, T.~M. and Colucci, S.~J. (1997).
\newblock Verification of {Eta--RSM} short-range ensemble forecasts.
\newblock {\em Monthly Weather Review}, 125:1312--1327.

\bibitem[Heinrich-Mertsching et~al., 2021]{Heinrich2021}
Heinrich-Mertsching, C., Thorarinsdottir, T.~L., Guttorp, P., and Schneider, M.
  (2021).
\newblock Validation of point process predictions with proper scoring rules.
\newblock {\em arXiv preprint arXiv:2110.11803}.

\bibitem[Henzi and Ziegel, 2022]{Henzi2022}
Henzi, A. and Ziegel, J.~F. (2022).
\newblock Valid sequential inference on probability forecast performance.
\newblock {\em Biometrika}, 109:647--663.

\bibitem[Hersbach et~al., 2020]{Hersbach2020}
Hersbach, H., Bell, B., Berrisford, P., Hirahara, S., Hor{\'a}nyi, A.,
  Mu{\~n}oz-Sabater, J., Nicolas, J., Peubey, C., Radu, R., Schepers, D.,
  et~al. (2020).
\newblock The {ERA5} global reanalysis.
\newblock {\em Quarterly Journal of the Royal Meteorological Society},
  146:1999--2049.

\bibitem[Kn{\"u}ppel et~al., 2022]{Knuppel2022}
Kn{\"u}ppel, M., Kr{\"u}ger, F., and Pohle, M.-O. (2022).
\newblock Score-based calibration testing for multivariate forecast
  distributions.
\newblock {\em arXiv preprint arXiv:2211.16362}.

\bibitem[Lakatos et~al., 2022]{Lakatos2022}
Lakatos, M., Lerch, S., Hemri, S., and Baran, S. (2022).
\newblock Comparison of multivariate post-processing methods using global ecmwf
  ensemble forecasts.
\newblock {\em arXiv preprint arXiv:2206.10237}.

\bibitem[Roberts and Lean, 2008]{Roberts2008}
Roberts, N.~M. and Lean, H.~W. (2008).
\newblock Scale-selective verification of rainfall accumulations from
  high-resolution forecasts of convective events.
\newblock {\em Monthly Weather Review}, 136:78--97.

\bibitem[Schefzik et~al., 2013]{Schefzik2013}
Schefzik, R., Thorarinsdottir, T.~L., and Gneiting, T. (2013).
\newblock Uncertainty quantification in complex simulation models using
  ensemble copula coupling.
\newblock {\em Statistical Science}, pages 616--640.

\bibitem[Scheuerer and Hamill, 2015]{Scheuerer2015}
Scheuerer, M. and Hamill, T.~M. (2015).
\newblock Variogram-based proper scoring rules for probabilistic forecasts of
  multivariate quantities.
\newblock {\em Monthly Weather Review}, 143:1321--1334.

\bibitem[Scheuerer and Hamill, 2018]{Scheuerer2018}
Scheuerer, M. and Hamill, T.~M. (2018).
\newblock Generating calibrated ensembles of physically realistic,
  high-resolution precipitation forecast fields based on {GEFS} model output.
\newblock {\em Journal of Hydrometeorology}, 19:1651--1670.

\bibitem[Scheuerer and M{\"o}ller, 2015]{Scheuerer2015b}
Scheuerer, M. and M{\"o}ller, D. (2015).
\newblock Probabilistic wind speed forecasting on a grid based on ensemble
  model output statistics.
\newblock {\em The Annals of Applied Statistics}, pages 1328--1349.

\bibitem[Smith and Hansen, 2004]{Smith2004}
Smith, L.~A. and Hansen, J.~A. (2004).
\newblock Extending the limits of ensemble forecast verification with the
  minimum spanning tree.
\newblock {\em Monthly Weather Review}, 132:1522--1528.

\bibitem[Talagrand, 1997]{Talagrand1997}
Talagrand, O. (1997).
\newblock Evaluation of probabilistic prediction systems.
\newblock In {\em Proc. ECMWF Workshop on predictability, ECMWF, Reading, UK}.

\bibitem[Thorarinsdottir et~al., 2016]{Thorarinsdottir2016}
Thorarinsdottir, T.~L., Scheuerer, M., and Heinz, C. (2016).
\newblock Assessing the calibration of high-dimensional ensemble forecasts
  using rank histograms.
\newblock {\em Journal of Computational and Graphical Statistics}, 25:105--122.

\bibitem[Thorarinsdottir and Schuhen, 2018]{Thorarinsdottir2018}
Thorarinsdottir, T.~L. and Schuhen, N. (2018).
\newblock Verification: assessment of calibration and accuracy.
\newblock In {\em Statistical postprocessing of ensemble forecasts}, pages
  155--186. Elsevier.

\bibitem[Vannitsem et~al., 2018]{Vannitsem2018}
Vannitsem, S., Wilks, D.~S., and Messner, J. (2018).
\newblock {\em Statistical postprocessing of ensemble forecasts}.
\newblock Elsevier.

\bibitem[Vovk and Wang, 2021]{Vovk2021}
Vovk, V. and Wang, R. (2021).
\newblock E-values: Calibration, combination and applications.
\newblock {\em The Annals of Statistics}, 49:1736--1754.

\bibitem[Waudby-Smith and Ramdas, 2023]{Waudby-SmithRamdas2023}
Waudby-Smith, I. and Ramdas, A. (2023).
\newblock Estimating means of bounded random variables by betting.
\newblock {\em Journal of the Royal Statistical Society: Series B}.
\newblock To appear as Discussion Paper. Preprint available at
  \url{arXiv:2010.09686}.

\bibitem[Weller and Hoeting, 2016]{Weller2016}
Weller, Z.~D. and Hoeting, J.~A. (2016).
\newblock A review of nonparametric hypothesis tests of isotropy properties in
  spatial data.
\newblock {\em Statistical Science}, pages 305--324.

\bibitem[Weniger et~al., 2017]{Weniger2017}
Weniger, M., Kapp, F., and Friederichs, P. (2017).
\newblock Spatial verification using wavelet transforms: a review.
\newblock {\em Quarterly Journal of the Royal Meteorological Society},
  143:120--136.

\bibitem[Wilks, 2019]{Wilks2019}
Wilks, D. (2019).
\newblock Indices of rank histogram flatness and their sampling properties.
\newblock {\em Monthly Weather Review}, 147:763--769.

\bibitem[Wilks, 2004]{Wilks2004}
Wilks, D.~S. (2004).
\newblock The minimum spanning tree histogram as a verification tool for
  multidimensional ensemble forecasts.
\newblock {\em Monthly Weather Review}, 132:1329--1340.

\bibitem[Wilks, 2017]{Wilks2017}
Wilks, D.~S. (2017).
\newblock On assessing calibration of multivariate ensemble forecasts.
\newblock {\em Quarterly Journal of the Royal Meteorological Society},
  143:164--172.

\bibitem[Ziegel, 2017]{Ziegel2017}
Ziegel, J. (2017).
\newblock Copula calibration.
\newblock In {\em Copulae: On the Crossroads of Mathematics and Economics},
  pages 7--10. Mathematisches Forschungsinstitut Oberwolfach. Report No.
  20/2015.

\end{thebibliography}

\section*{Appendix}

\subsection*{Continuous probabilistic forecast distributions}

Suppose we are interested in forecasting a univariate random variable $Y$, and let $F$ be a forecast for $Y$ in the form of a cumulative distribution function. The forecast is said to be probabilistically calibrated if its probability integral transform (PIT)
\begin{equation*}
	Z_{F} = F(Y-) + V[F(Y) - F(Y-)] 
\end{equation*}
follows a standard uniform distribution, where $V$ is a standard uniform random variable that is independent from $Y$ and $F$, and $F(Y-) = \mathrm{lim}_{x \uparrow Y} F(x)$ \citep{Diebold1998,Gneiting2013}. The motivation behind this definition is the result that $Y \sim F$ if and only if $Z_{F}$ follows a standard uniform distribution. The probabilistic calibration of a forecast system (rather than a single forecast) can be assessed by treating the forecast $F$ as random.

\bigskip

A stronger notion of calibration for a forecast system is auto-calibration, $\mathcal{L}(Y \mid F) = F$, where $\mathcal{L}$ denotes the law, or distribution. That is, the conditional distribution of $Y$ given the (random) forecast $F$ is equal to $F$. Auto-calibration is a strictly stronger requirement than probabilistic calibration \citep{Gneiting2013,Gneiting2021}. However, in contrast to probabilistic calibration, auto-calibration generalises without any adaptations to multivariate predictions and outcomes.

\bigskip

In practice, the probabilistic calibration of a forecast system can be assessed by calculating the PIT values corresponding to a sequence of forecasts and observations. These PIT values can then be displayed in a histogram to check whether or not they resemble a sample from a standard uniform distribution. These PIT histograms generalise rank histograms, and their interpretation is analogous to the interpretation of rank histograms discussed in Section \ref{section:histograms}. 

\bigskip

Suppose now that we wish to predict a multivariate random variable $\mathbf{Y}$ that takes values in $\R^{d}$, for $d>1$. Given a simple pre-rank function $\rho: \R^{d} \to \R$ and a multivariate forecast distribution $F$, $F$ is said to be probabilistically calibrated with respect to the pre-rank function $\rho$ if 
\begin{equation}\label{eq:mvPIT}
	Z_{F_{\rho}} = F_{\rho}(\rho(\mathbf{Y})-) + V[F_{\rho}(\rho(\mathbf{Y})) - F_{\rho}(\rho(\mathbf{Y})-)] 
\end{equation}
follows a standard uniform distribution, where $V$ is a standard uniform random variable, and $F_{\rho}$ is the univariate cumulative distribution function induced by $F$ and $\rho$, i.e. $F_{\rho}(x) = \PP(\rho(\mathbf{X}) \leq x)$ for $\mathbf{X} \sim F$ and $x \in \R$. This definition of multivariate calibration is equivalent to assessing whether the forecast distribution is probabilistically calibrated when predicting the univariate random variable $\rho(\mathbf{Y})$.
The following simple but powerful result gives a theoretical justification for considering probabilistic calibration with respect to pre-rank functions.
\begin{theorem}
	Suppose that the multivariate forecast distribution $F$ is auto-calibrated for $\mathbf{Y} \in \R^d$. Then $F$ is probabilistically calibrated with respect to any simple pre-rank function $\rho:\R^d \to \R$. 
\end{theorem}
\begin{proof}
	Let $x \in [0,1]$. Then
	\[
	\PP(Z_{F_\rho} \le x) = \mathbb{E}\Big[\mathbb{E}\big[\one\{Z_{F_\rho} \le x\} \mid F \big]\Big] 
	= \mathbb{E}[x] = x,
	\]
	where we used auto-calibration for the second equality.
\end{proof}

\bigskip

Multivariate PIT histograms can analogously be defined to assess the calibration of continuous multivariate forecast distributions. The pre-rank functions listed in Section \ref{section:existing_pre-ranks} can equally be applied to continuous multivariate forecast distributions but the distribution $F_\rho$ may have to be approximated by simulation. 

\bigskip

\cite{Ziegel2017} demonstrates how pre-rank functions for ensemble forecasts depending on $M+1$ arguments (instead of just the first argument) can be generalised to the case where the forecast is an arbitrary predictive distribution over $\R^{d}$. We do not detail the arguments here since we advocate simple pre-rank functions depending on the first argument only.

\end{document}